\pdfoutput=1
\documentclass[12pt, letterpaper]{article}

\usepackage[margin=1in]{geometry}
\usepackage{amsmath, amssymb, amsfonts}
\usepackage{amsthm}
\usepackage{natbib}
\usepackage{graphicx}
\usepackage[utf8]{inputenc}
\usepackage[hidelinks]{hyperref}

\newtheorem{theorem}{Theorem}

\title{Meta-Analysis Without Normality: Estimating the True Effect Distribution with Penalized Gaussian Mixtures}
\author{
  Daihe Sui and Elizabeth Tipton\\
  \vspace{0.2cm}
  \normalsize Department of Statistics and Data Science, Northwestern University
}
\date{}

\begin{document}

\maketitle

\begin{abstract}
Standard random-effects meta-analysis relies heavily on the assumption that the underlying true effects are normally distributed. In the social sciences, where evidence synthesis increasingly involves large, highly heterogeneous datasets, this assumption is often restrictive and unjustified. Misspecification of the random-effects distribution prevents the detection of asymmetry or multimodality, potentially leading to erroneous conclusions regarding the prevalence of adverse effects or the existence of specific subgroups. This paper introduces a Penalized Gaussian Mixture (PGM) framework designed to recover the entire probability density function of true effects without enforcing a rigid parametric shape. The method adapts to different non-normal scenarios, including skewed and multimodal distributions, while reducing to the normal case when supported by the data. A simulation study demonstrates that in large, highly heterogeneous meta-analyses, PGM yields substantially more accurate estimates of tail probabilities and the density function than standard methods when normality is violated, without substantially compromising efficiency under normality. An empirical application to environmental education data illustrates the practical utility of the method. The proposed framework provides researchers with a robust tool to move beyond simple summary statistics and characterize the complex nature of the true effect distribution in the real world.

\vspace{0.5cm}
\end{abstract}

\section{Introduction}

Meta-analysis is a quantitative, systematic approach for synthesizing results from multiple empirical studies that address a common research question. The term was coined by Gene V. Glass in his presidential address to the American Educational Research Association~\citep{Glass1976}, and the methodology was formalized by~\citet{Hedges1985}. By combining data across independent investigations, meta-analysis increases statistical power, refines effect size estimates, reconciles contradictory findings, and enables the exploration of moderators through subgroup analyses and meta-regression. Consequently, meta-analytic methods have become indispensable in fields such as psychology, education, medicine, and public health, underpinning both theoretical advancement and evidence-based practice.

Unlike in medicine, where meta-analyses often maintain a strict, narrow focus and rely on smaller samples of clinical trials, evidence synthesis in the social sciences typically presents a different profile. A recent review of 1,000 meta-analyses across 10 different disciplines shows that social science meta-analyses are frequently (i) much larger, often encompassing up to hundreds of studies, and (ii) substantially more heterogeneous, as they tend to include a broad spectrum of operational definitions, settings, and designs~\citep{Wu2025}. Given these two characteristics, the study of between-study variation becomes paramount, as the sheer volume and diversity of data allow for a more rigorous examination of how effects differ across contexts. This investigation into the variability between studies, known as heterogeneity, is often the most insightful part of a social science meta-analysis. It helps us shift from asking a simple question: ``Does it work?'' to a far more challenging, but meaningful one: ``Under what circumstances does it work best?''

In the presence of heterogeneity, the random-effects model is typically the method of choice. The conventional framework usually assumes a normal-normal model: each observed effect is normally distributed around its true effect, and the underlying true effects are also normally distributed, explicitly modeling variability across studies. While one might argue that some method of moments estimators do not strictly require the latter normality assumption, statistical inference often relies on it. The former normality assumption is well justified by the Central Limit Theorem; however, the latter lacks similar theoretical justification. In fact, histograms of observed effects in some meta-analyses display obvious non-normality. Since observed effects are effectively a convolution of the true effects and normal noise, this suggests that the distribution of the true effects is likely even more non-normal in those cases. Researchers often use meta-regression to attempt to explain heterogeneity using study characteristics or subgroup information. However, meta-regression has its limits; it is effective only when the covariates adequately account for the variability of effects. Frequently, substantial residual heterogeneity persists after adjustment, sometimes still appearing as complex patterns. In these instances, the normality assumption becomes dubious. Yet, despite these potential shortcomings, the vast majority of meta-analyses continue to model between-study variability using a normal distribution. This reliance is largely due to analytical convenience, model simplicity, and software availability.

One might naturally ask whether misspecifying the distribution of true effects as normal, when it is actually non-normal, yields detrimental consequences. Some simulation studies, such as the work by~\citet{Kontopantelis2010}, have indicated that the performance of standard random-effects models remains quite robust against even severe deviations from normality. However, it is crucial to note that such findings primarily concern the estimation of the mean effect. As previously noted, in the context of large, heterogeneous meta-analyses common in the social sciences, a single summary statistic for the mean is rarely sufficient to fully describe the distribution of true effects. In these cases, the heterogeneity itself is of equal, if not greater, importance. The prediction interval has gained substantial popularity as a tool to interpret heterogeneity when it exists~\citep{Riley2011, Borenstein2017}; however, its validity depends directly on the correctness of the underlying distributional assumption. Under the strict normal assumption, heterogeneity is entirely captured by a single parameter: the variance of the true effects distribution ($\tau^2$). However, if we allow for non-normality, variance alone becomes an inadequate descriptor. To truly understand the nature of the evidence, one requires insight into the entire shape of the true effects distribution.

To visualize why moments alone are insufficient, consider the illustrative example in Figure~\ref{fig:distributions}. The plot displays three distinct probability density functions: a normal distribution, a shifted log-normal distribution, and a mixture of two normal distributions. Crucially, all three distributions are constructed to have identical first and second moments (mean = 1, variance = 1). Despite these identical summary statistics, they convey drastically different substantive narratives. If these densities represented the distribution of true effects for a proposed intervention, the policy implications would diverge sharply. Compared to the normal baseline (solid blue), the log-normal distribution (dashed red) suggests a somewhat safer profile, where the probability of an adverse effect ($P(\theta<0)$) is 0.10. Conversely, the mixture distribution (dotted green) reveals a polarized reality: while the mean remains positive, a distinct sub-population of studies experiences significant harm ($P(\theta<0) = 0.27$). A decision framework based solely on the mean and variance would fail to distinguish the relatively safe log-normal scenario from the highly risky mixture scenario, potentially leading to the adoption of a treatment that is detrimental in a non-negligible number of scenarios.

\begin{figure}[htbp]
    \centering
    \includegraphics[width=0.9\textwidth]{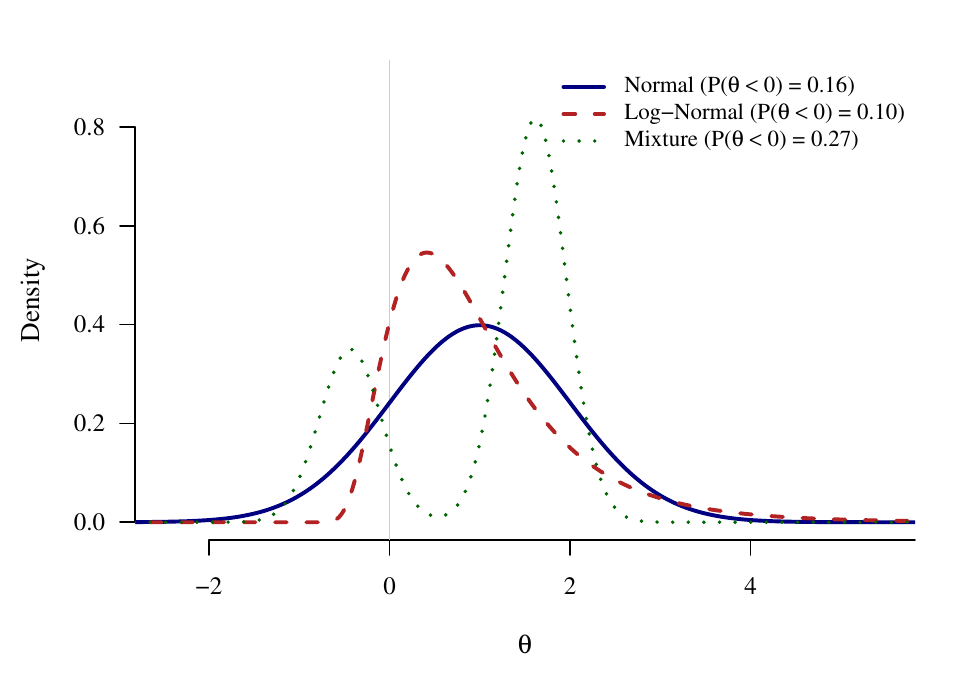}
    \caption{Three distinct distributions (Normal, Log-Normal, and Mixture) that share the same mean ($\mu=1$) and variance ($\tau^2=1$). Despite identical first and second moments, the shapes differ drastically, leading to significantly different probabilities of adverse effects ($P(\theta<0)$).}
    \label{fig:distributions}
\end{figure}

Consequently, many researchers have proposed methods that relax the normality assumption. A recent review by~\citet{Panagiotopoulou2025} examined 24 alternative meta-analytic models. However, most of these models address only specific types of deviations. Essentially, these approaches merely substitute the normality assumption with another specific constraint, which often remains insufficiently flexible. In this paper, we propose a new method designed to handle diverse distributional patterns, including skewed and multimodal distributions, while maintaining performance not much worse than the standard normal-normal model when the true effects are, in fact, normally distributed. In short, we introduce a robust and universally applicable framework.

To achieve this, we propose using a technique called Penalized Gaussian Mixture (PGM) to model the random-effects distribution. This approach utilizes a preset grid of normal densities as building blocks to ``build'' distributions of any shape, while imposing a strong penalty for smoothness to avoid overfitting. The conceptual motivation dates back to~\citet{Eilers1996}, who used B-splines rather than normal densities as the building blocks in regression and density estimation. Similar ideas were subsequently applied in papers such as~\citet{Ghidey2004} and~\citet{Komarek2008}. We illustrate the derivation of comprehensive distributional properties using PGM, including the entire distribution function, moving beyond the simple mean and variance of a normal model.

\section{Method}

\subsection{Conventional Random-Effects Meta-Analysis}

The standard random-effects model provides a foundational framework for meta-analysis. It rests on a two-level hierarchical structure and several key assumptions. At the first level, it is assumed that each of the $N$ independent studies provides an observed effect size, $y_i$, that is an unbiased estimate of a study-specific true effect, $\theta_i$. The sampling error is modeled as normally distributed with a known within-study variance, $v_i$:
\[y_i \mid \theta_i \sim \mathcal{N}(\theta_i, v_i) \quad \text{for } i = 1, \dots, N.\]
Crucially, the within-study variances, $v_i$, are treated as known fixed constants. In practice, they are estimates derived from the primary studies, making this an approximation that works well when primary studies have large sample sizes but can be problematic otherwise.

At the second level, the true effects $\theta_i$ are assumed to be a random sample from a normally distributed superpopulation of effects. This distribution is characterized by a mean $\mu$, representing the overall average true effect, and a variance $\tau^2$, known as the between-study variance or heterogeneity:
\[\theta_i \sim \mathcal{N}(\mu, \tau^2).\]
The parameter $\tau^2$ is a critical quantity, representing the extent to which true effects vary across studies. A value of $\tau^2 = 0$ implies that all studies share a single common true effect, reducing the model to a fixed-effect (or common-effect) model.

Integrating out the latent true effects $\theta_i$ yields the marginal distribution of the observed effect sizes, which is also normal:
\[y_i \sim \mathcal{N}(\mu, v_i + \tau^2).\]
The goal of the meta-analysis is to estimate the two primary parameters: the overall mean effect, $\mu$, and the heterogeneity variance, $\tau^2$. Estimation typically proceeds in two stages. First, $\tau^2$ is estimated, and many methods exist for this purpose~\citep{Veroniki2016}. The traditional yet popular method-of-moments estimator of~\citet{DerSimonian1986} (DL), which is based on Cochran's Q statistic, is considered distribution-free, as it only requires assumptions about the existence of the first two moments of the distribution of effects. Likelihood-based approaches such as maximum likelihood (ML) and restricted maximum likelihood (REML) are obtained by maximizing the log-likelihood function (and restricted log-likelihood) from the marginal normal distribution of the $y_i$, and thus explicitly hinge on the assumption that the true effects are normally distributed.

Once an estimate $\hat{\tau}^2$ is obtained, the overall mean $\mu$ is estimated via a weighted average:
\[\hat{\mu} = \frac{\sum_{i=1}^N w_i y_i}{\sum_{i=1}^N w_i}, \quad \text{with weights } w_i = \frac{1}{v_i + \hat{\tau}^2}.\]
The variance of this estimate, $V(\hat{\mu}) = 1/\sum_{i=1}^N w_i$, is used to construct a confidence interval for $\mu$. When employing likelihood-based methods (ML or REML), statistical inference extends beyond point estimation for heterogeneity. Specifically, one can simply derive the standard error for the heterogeneity estimate, $\hat{\tau}^2$, from the inverse of the Fisher Information Matrix, allowing for the construction of a Wald-type confidence interval. Another more complex alternative is the profile likelihood method~\citep{Hardy1996}, which inverts the likelihood ratio test to generate a confidence interval that respects the parameter space and sampling distribution asymmetry.

The normal-normal model is widely used due to parsimony and tractability. However, its validity hinges on the strong assumption that the true effects $\theta_i$ are normally distributed. The potential violation of this assumption is the primary motivation for the flexible approach proposed next.

\subsection{Penalized Gaussian Mixture}

Instead of relying on the strict normality assumption or substituting it with another specific parametric family, we propose using a Penalized Gaussian Mixture (PGM). This method offers a flexible framework capable of accommodating a wide variety of distributional shapes for the true effects.

The core intuition behind the PGM is best understood visually before detailing the mathematical formulation. As illustrated in Figure~\ref{fig:method}, the approximated probability density function (represented by the dashed black line) is constructed as a weighted sum of multiple underlying ``building blocks.'' These building blocks are fixed normal densities (represented by the dotted grey lines), which are shifted along an evenly spaced grid covering the range of the effects.

\begin{figure}[htbp]
    \centering
    \includegraphics[width=0.85\textwidth, height=0.85\textheight, keepaspectratio]{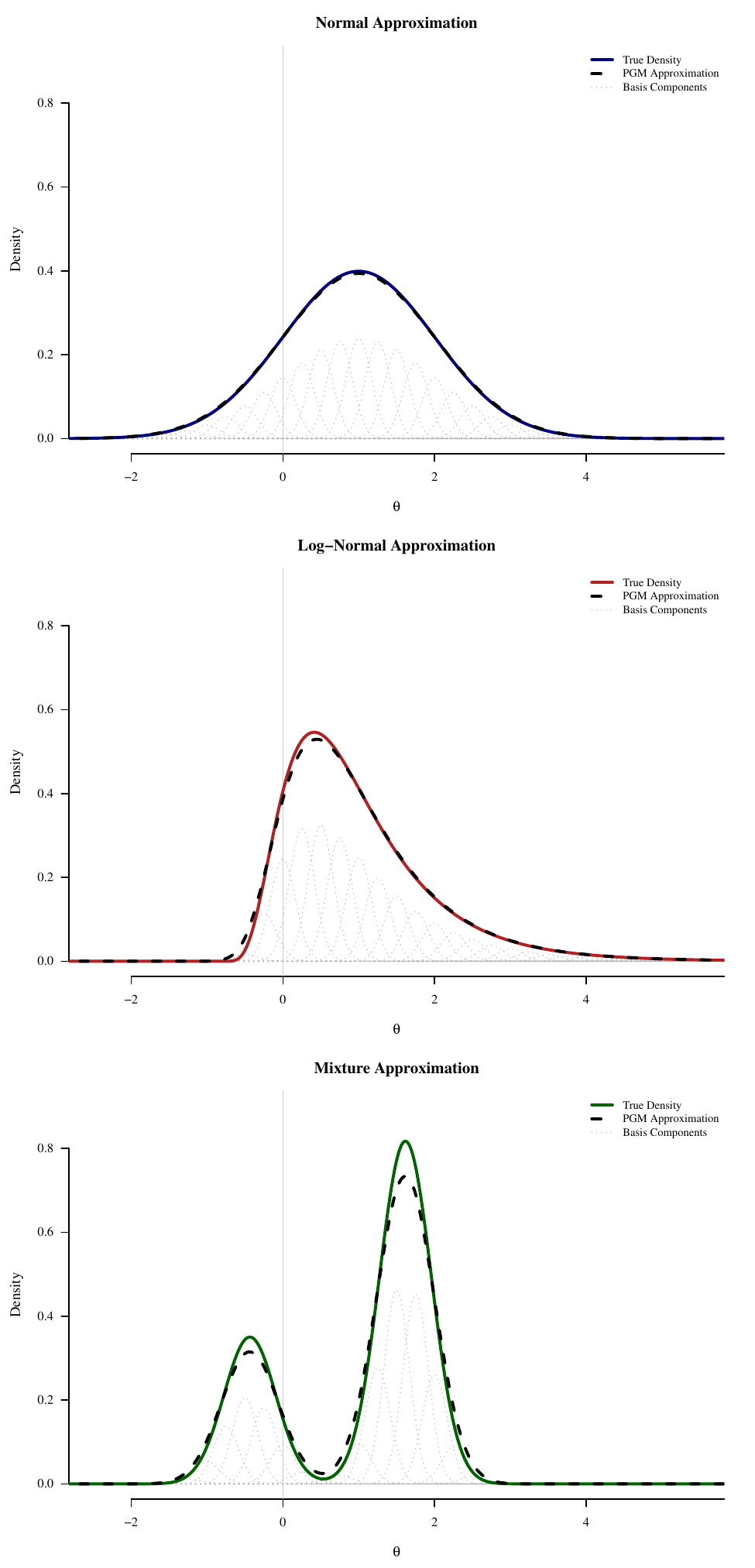}
    \caption{Visual representation of the Penalized Gaussian Mixture (PGM) approach approximating the three distributions introduced in Figure~\ref{fig:distributions}. The solid colored lines represent the true underlying densities, while the dashed black lines represent the PGM approximations. These complex shapes are constructed by aggregating the underlying basis functions (dotted grey lines), which act as weighted building blocks.}
    \label{fig:method}
\end{figure}

By assigning different weights to these component densities, the model can approximate complex patterns that standard methods miss. To demonstrate this versatility, Figure~\ref{fig:method} revisits the three exact distributions introduced earlier: the Normal, Log-Normal, and Mixture scenarios. As shown in the three panels, the PGM framework successfully approximates each distinct shape. In the Normal case, the weights are symmetric and bell-shaped, perfectly reconstructing the standard Gaussian curve. For the Log-Normal distribution, higher weights are assigned asymmetrically to components on the left, while extending a long tail to the right, effectively capturing the skew. Most notably, in the bimodal Mixture scenario, two distinct clusters of higher weights naturally form a distribution with two peaks, revealing the polarized sub-populations. By comparing the true underlying densities (solid colored lines) with the PGM estimates (dashed black lines), it is evident that rather than forcing the data to fit a rigid, pre-defined parametric family, the PGM approach allows these basic building blocks to adaptively form the shape dictated by the data. This flexibility ensures that critical substantive details---such as the elevated probability of adverse effects seen in the mixture scenario---are captured rather than obscured.

Formally, the model is defined for effect sizes $i=1, \dots, N$ as:

\begin{align*}
    y_i \mid \theta_i &\sim \mathcal{N}(\theta_i, v_i)\\
    \theta_i &\sim \sum_{k=1}^K w_k \mathcal{N}(\mu_k, \tau_c^2).
\end{align*}
The variance terms $v_i$ are known. The parameters $\mu_k$ for $k=1, \dots, K$ are preset on an evenly-spaced grid with spacing $\Delta$. The common standard deviation of the basis densities is defined as $\tau_c = c \cdot \Delta$, where $c$ is a fixed scaling constant.

It is crucial to distinguish this model from a ``real'' Gaussian Mixture Model (GMM) for the distribution of $\theta_i$. In a standard GMM, the means and variances of the components are free parameters estimated to identify latent subpopulations. In the PGM approach, however, the component means and standard deviations are fixed constants. Consequently, the individual normal components do not represent distinct, meaningful subgroups; rather, they serve merely as mathematical building blocks (basis functions). The flexibility of the model arises entirely from the estimation of the weights, allowing these fixed blocks to aggregate into a smooth density of any shape.

The choice of the grid boundaries and the constant $c$ requires careful consideration. First, regarding the boundaries of the means $[\mu_1, \mu_K]$, we propose a robust, data-driven approach to determine the appropriate limits based on the observed effect sizes. We compute a robust overall mean ($\hat{\mu}_{\text{rob}}$) using the median of the observed effects, and a robust between-study variance defined as $\hat{\tau}^2_{\text{rob}} = \text{MAD}^2(y_i) - \text{median}(v_i)$, where MAD is the Median Absolute Deviation. We then calculate robust empirical shrinkage estimates of the true effects for each study, defined as $\tilde{\theta}_i = \hat{\mu}_{\text{rob}} + \frac{\hat{\tau}^2_{\text{rob}}}{\hat{\tau}^2_{\text{rob}} + v_i}(y_i - \hat{\mu}_{\text{rob}})$. The grid range is constructed by taking the minimum and maximum of these shrunk estimates, expanded outward by a padding factor of $1.5\hat{\tau}_{\text{rob}}$ on both ends (or $1.5\sqrt{\text{median}(v_i)}$ if $\hat{\tau}^2_{\text{rob}} \le 0$). This adaptive boundary ensures the grid spans the empirical support of the data while resisting the influence of extreme outliers that could unnecessarily stretch the grid. Regardless of the method, the choice involves a trade-off: one does not want the range to be too large or too small. If the range is too large, more components ($K$) are required to maintain the same resolution (holding $\Delta$ constant), which increases the computational burden. Conversely, if the range is too small, the model will miss a significant amount of probability mass or density. However, based on our experience, it is generally better to be conservative and specify a larger range rather than using a smaller range and risking the omission of important chunks of the underlying density.

Likewise, the common standard deviation $\tau_c$ (and therefore, the constant $c$) must be chosen strategically. We recommend setting the common SD such that $c = 2/3$ (i.e., $\tau_c = \frac{2}{3}\Delta$), which mimics the behavior of a cubic B-spline basis whose knots are on an equidistant grid. Again, $c$ cannot be excessively small or large. If $c$ is too large, resolution is lost due to oversmoothing. If $c$ is too small, the sum of the Gaussian densities does not approximate unity uniformly across the domain, which leads to non-negligible bias.

Similar consideration must be given to the choice of $K$, which determines the number of mixture components. In the spirit of the original use of B-spline building blocks~\citep{Eilers1996}, $K$ is usually chosen to be sufficiently large. The philosophy here is to allow the model sufficient flexibility to capture complex features, relying on the strong difference penalty to smooth the weights and prevent overfitting, rather than restricting the model capacity a priori. In practice, however, there is little practical gain from choosing an extremely large $K$. For a meta-analysis with a moderate sample size, we can typically choose $K=20$ or $K=30$. This provides sufficient resolution to capture the shape of the distribution without incurring unnecessary computational costs.

The marginal distribution of $y_i$ is a mixture of normals:
\[y_i \sim \sum_{k=1}^K w_k \mathcal{N}(\mu_k, \tau_c^2 + v_i).\]
The full vector of mixture weights $\boldsymbol{w}$ ($K \times 1$) is reparametrized using a softmax transformation of baseline parameters $\boldsymbol{\alpha}^* = (\alpha^*_1, \dots, \alpha^*_K)^\top$:
\[w_k = \frac{\exp(\alpha^*_k)}{\sum_{l=1}^K \exp(\alpha^*_l)}.\]
For identifiability, we enforce a constraint on the center reference component, $k_0 = \lfloor (K+1)/2 \rfloor$, such that $\alpha^*_{k_0} = 0$. The free parameters are denoted by the vector $\boldsymbol{\alpha}$ ($(K-1) \times 1$). To facilitate concise matrix operations, we define $\boldsymbol{C}$ as the $K \times (K-1)$ constraint matrix that maps the free parameters to the full vector ($\boldsymbol{\alpha}^* = \boldsymbol{C} \boldsymbol{\alpha}$) by inserting a zero at index $k_0$.

Let $\phi_{ik} = \phi(y_i \mid \mu_k, \tau_c^2 + v_i)$ denote the component density, and $f(y_i) = \sum_{k=1}^K w_k \phi_{ik}$ the marginal density. The log-likelihood of the observed data is then $l(\boldsymbol{\alpha}) = \sum_{i=1}^N \log f(y_i)$.

Maximizing this likelihood with a large number of basis functions risks overfitting. To enforce smoothness, we employ a penalized log-likelihood approach with a difference penalty on the transformed weights. The objective function becomes:
\[l_p(\boldsymbol{\alpha} \mid \lambda_\alpha) = \sum_{i=1}^N \log\left( \sum_{k=1}^K w_k \phi_{ik} \right) - \frac{\lambda_\alpha}{2} \boldsymbol{\alpha}^\top \boldsymbol{C}^\top \boldsymbol{D}^\top \boldsymbol{D} \boldsymbol{C} \boldsymbol{\alpha},\]
where $\boldsymbol{D}$ is a difference matrix of order $d$. In practice, $d$ is typically chosen to be 2 or 3. Specifically, when $d=3$, the high penalty forces the third-order differences of the coefficients to zero, resulting in a Gaussian shape. This effectively reduces the model to the normal-normal model in the limiting scenario, which is a highly desirable property.

Since no analytical solution exists for maximizing $l_p(\boldsymbol{\alpha} \mid \lambda_\alpha)$, we employ the Newton-Raphson method. This requires deriving the gradient (score vector) and the negative Hessian (Observed Information Matrix). We define the matrix $\boldsymbol{Q}$ ($N \times K$) with elements $q_{ik} = \frac{w_k \phi_{ik}}{f(y_i)}$.

\begin{theorem}[Penalized Score and Observed Information]
The penalized score vector $\nabla_{\alpha} l_p$ is given by:
\[\nabla_{\alpha} l_p = \boldsymbol{C}^\top \left[ \boldsymbol{Q}^\top \boldsymbol{1}_N - N \boldsymbol{w} - \lambda_\alpha \boldsymbol{D}^\top \boldsymbol{D} \boldsymbol{C} \boldsymbol{\alpha} \right].\]
The penalized Observed Information Matrix (the negative Hessian, $-\nabla^2_{\alpha\alpha} l_p$) is:
\[-\nabla^2_{\alpha\alpha} l_p = \boldsymbol{C}^\top \left[ \boldsymbol{Q}^\top \boldsymbol{Q} - N \boldsymbol{w} \boldsymbol{w}^\top - \operatorname{diag}(\boldsymbol{Q}^\top \boldsymbol{1}_N - N \boldsymbol{w}) + \lambda_\alpha \boldsymbol{D}^\top \boldsymbol{D} \right] \boldsymbol{C}.\]
\end{theorem}
\begin{proof}
The log-likelihood contribution of study $i$ is $l_i = \log \sum_{k=1}^K w_k \phi_{ik}$. The derivative with respect to the transformed weights $\boldsymbol{\alpha}^*$ is obtained by applying the chain rule through $\boldsymbol{w}$. Let $\delta_{kj}$ denote the Kronecker delta (equal to 1 if $k=j$ and 0 otherwise). Given the softmax function, $\frac{\partial w_k}{\partial \alpha^*_j} = w_k (\delta_{kj} - w_j)$. Thus, $\frac{\partial l_i}{\partial \alpha^*_j} = \sum_k \frac{\phi_{ik}}{f(y_i)} w_k (\delta_{kj} - w_j) = \sum_k q_{ik}(\delta_{kj} - w_j) = q_{ij} - w_j$. Let $\boldsymbol{q}_i = (q_{i1}, \dots, q_{iK})^\top$ denote the $K \times 1$ column vector corresponding to the $i$-th row of $\boldsymbol{Q}$. Then in vector form, $\nabla_{\alpha^*} l_i = \boldsymbol{q}_i - \boldsymbol{w}$. Summing over all $N$ studies yields $\boldsymbol{Q}^\top \boldsymbol{1}_N - N\boldsymbol{w}$. Applying the linear constraint $\boldsymbol{\alpha}^* = \boldsymbol{C}\boldsymbol{\alpha}$ gives the unpenalized score $\boldsymbol{C}^\top (\boldsymbol{Q}^\top \boldsymbol{1}_N - N\boldsymbol{w})$. Adding the derivative of the penalty term $-\frac{\lambda_\alpha}{2} \boldsymbol{\alpha}^\top \boldsymbol{C}^\top \boldsymbol{D}^\top \boldsymbol{D} \boldsymbol{C} \boldsymbol{\alpha}$ directly produces the penalized score.

For the Hessian, we differentiate $\boldsymbol{q}_i - \boldsymbol{w}$ with respect to $\boldsymbol{\alpha}^*$. The Jacobian of $\boldsymbol{q}_i$ is $\operatorname{diag}(\boldsymbol{q}_i) - \boldsymbol{q}_i \boldsymbol{q}_i^\top$, and the Jacobian of $\boldsymbol{w}$ is $\operatorname{diag}(\boldsymbol{w}) - \boldsymbol{w} \boldsymbol{w}^\top$. Therefore, the second derivative for study $i$ is $\operatorname{diag}(\boldsymbol{q}_i - \boldsymbol{w}) - \boldsymbol{q}_i \boldsymbol{q}_i^\top + \boldsymbol{w} \boldsymbol{w}^\top$. Summing over $i$ and taking the negative yields the unpenalized Observed Information: $\boldsymbol{Q}^\top \boldsymbol{Q} - N\boldsymbol{w} \boldsymbol{w}^\top - \operatorname{diag}(\boldsymbol{Q}^\top \boldsymbol{1}_N - N\boldsymbol{w})$. Including the second derivative of the penalty term ($\lambda_\alpha \boldsymbol{D}^\top \boldsymbol{D}$) and pre/post-multiplying by $\boldsymbol{C}^\top$ and $\boldsymbol{C}$ respectively completes the proof.
\end{proof}

The smoothing parameter $\lambda_\alpha$ is selected to minimize the Akaike Information Criterion (AIC). Because the penalty term effectively restricts the freedom of the parameters, we must utilize the effective degrees of freedom (EDF) derived from the information matrices. Let $\boldsymbol{J}_p = -\nabla^2_{\alpha\alpha} l_p(\hat{\boldsymbol{\alpha}} \mid \lambda_\alpha)$ denote the penalized observed information matrix evaluated at the estimates, and let $\boldsymbol{J} = -\nabla^2_{\alpha\alpha} l(\hat{\boldsymbol{\alpha}})$ denote the unpenalized observed information matrix. Consequently, the AIC is calculated as:
\[\text{AIC}(\lambda_\alpha) = -2l(\hat{\boldsymbol{\alpha}}) + 2\operatorname{tr}(\boldsymbol{J}_p^{-1}\boldsymbol{J}).\]
The trace term, $\operatorname{tr}(\boldsymbol{J}_p^{-1}\boldsymbol{J})$, represents the effective degrees of freedom, quantifying the complexity of the fitted model for a given smoothing parameter $\lambda_\alpha$. Notably, in the unpenalized scenario where $\lambda_\alpha = 0$, the penalized information matrix $\boldsymbol{J}_p$ is identical to the unpenalized matrix $\boldsymbol{J}$. In this case, the product $\boldsymbol{J}_p^{-1}\boldsymbol{J}$ reduces to the identity matrix, and the trace naturally simplifies to the exact number of free parameters in the model.

After obtaining the optimal smoothing parameter $\lambda_\alpha$ and the corresponding penalized maximum likelihood estimates $\hat{\boldsymbol{\alpha}}$, we can directly compute the estimated weights vector $\hat{\boldsymbol{w}}$. Because the proposed model is a mixture of known basis functions, key distributional properties are simply linear combinations or straightforward functions of these weights. For instance, the overall mean of the true effects is estimated as the weight-averaged sum of the component means, $\hat{\mu} = \boldsymbol{\mu}^\top \hat{\boldsymbol{w}}$. Similarly, the variance (heterogeneity) is obtained via the law of total variance as $\hat{\tau}^2 = \tau_c^2 + (\boldsymbol{\mu}^{(2)})^\top \hat{\boldsymbol{w}} - \hat{\mu}^2$, where $\boldsymbol{\mu}^{(2)}$ is the vector of squared component means.

Beyond basic moments, the primary advantage of the PGM framework lies in recovering the entire shape of the true effects distribution. The probability density function (PDF) at any given point $\theta^*$ is naturally evaluated as the weighted sum of the component densities, $\hat{f}(\theta^*) = \boldsymbol{\phi}(\theta^*)^\top \hat{\boldsymbol{w}}$, where $\boldsymbol{\phi}(\theta^*)$ is the vector of basis normal densities evaluated at $\theta^*$. By substituting these component densities with their cumulative counterparts, $\boldsymbol{\Phi}(\theta^*)$, we obtain the cumulative distribution function (CDF), $\hat{F}(\theta^*) = \boldsymbol{\Phi}(\theta^*)^\top \hat{\boldsymbol{w}}$. This formulation allows researchers to easily calculate tail probabilities, such as the proportion of adverse effects ($P(\theta < 0)$). Finally, quantiles can be estimated by simply finding the root that inverts the CDF, which is useful for constructing prediction intervals to interpret the spread of effects. If necessary, standard errors and confidence intervals for all of these derived quantities can be analytically obtained using the Delta method, utilizing the covariance matrix of the penalized parameter estimates ($\boldsymbol{V}_\alpha = \boldsymbol{J}_p^{-1} \boldsymbol{J} \boldsymbol{J}_p^{-1}$).

\subsection{Meta-Regression Extensions}

\subsubsection{Location-Shifting Meta-Regression Model}
\label{sec:location_meta_regression}

While the intercept-only PGM model provides a flexible approach to estimating the overall, unconditional distribution of true effects, a central goal in many syntheses is to explicitly explain observed heterogeneity rather than merely quantify it. To accommodate this, we now extend our framework into the meta-regression setting by incorporating covariates. The first approach we introduce to achieve this is the location-shifting model. In this framework, covariates uniformly shift the entire distribution's location along the effect-size continuum without altering its fundamental shape.

Formally, let $\boldsymbol{x}_i$ be the vector of $p$ covariates for study $i$, and $\boldsymbol{\beta}$ the corresponding fixed-effect coefficients. The location-shifting model is defined as:
\begin{align*}
y_i \mid \theta_i &\sim \mathcal{N}(\theta_i, v_i)\\
\theta_i &\sim \sum_{k=1}^K w_k \mathcal{N}(\mu_k + \boldsymbol{x}_i^\top \boldsymbol{\beta}, \tau_c^2)
\end{align*}
This formulation enforces the constraint that the mixture weights $w_k$ remain constant across all studies. Consequently, any change in the covariates simply translates the entire baseline density curve along the x-axis, preserving its variance, skewness, and overall shape.

This behavior is visually demonstrated in Figure~\ref{fig:location_shifting}. The baseline distribution (solid red line) represents the density of true effects when $x=0$. When the covariate changes to $x=-1$ (dotted blue line) or $x=1$ (dashed green line), the entire distribution simply shifts to the left or right along the continuum. Because the underlying component weights are static, the structural properties of the density---such as its variance and skewness---remain perfectly preserved across all contexts.

\begin{figure}[htbp]
    \centering
    \includegraphics[width=0.8\textwidth]{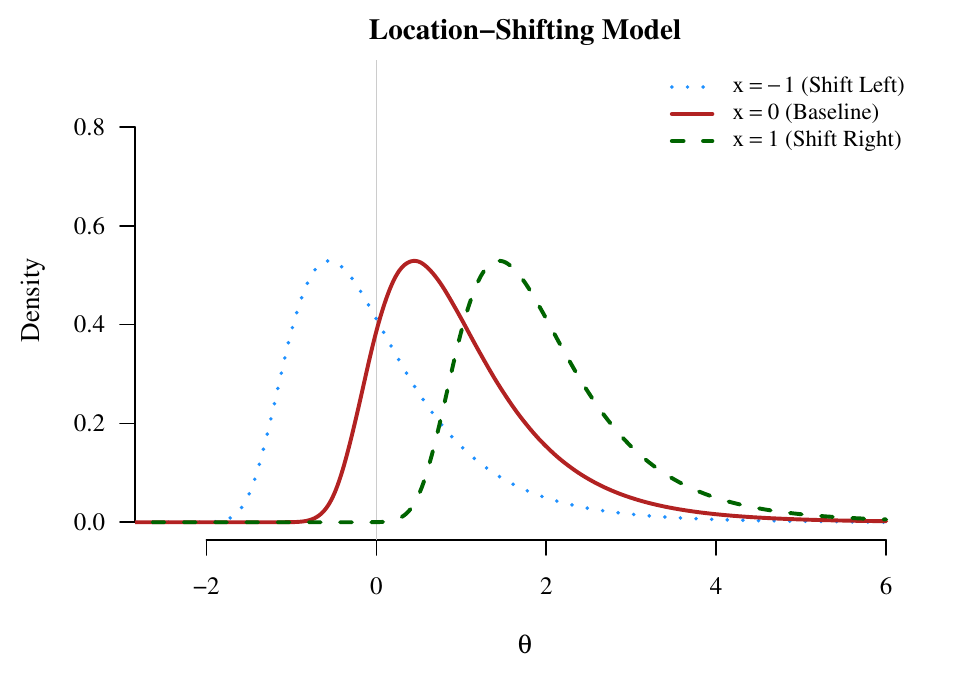}
    \caption{Illustration of the Location-Shifting Model. Covariates shift the entire distribution along the effect-size continuum (e.g., shifting left when $x=-1$ and right when $x=1$) while maintaining the exact same variance and shape as the baseline ($x=0$).}
    \label{fig:location_shifting}
\end{figure}

To define the baseline grid boundaries $[\mu_1, \mu_K]$ for the unshifted component means in the presence of covariates, we adapt the robust data-driven approach used in the intercept-only model. First, we fit a preliminary weighted least squares regression of $y_i$ on $\boldsymbol{x}_i$, weighted by the inverse of the known within-study variances ($1/v_i$), to partial out the linear effects of the covariates. This isolates the base residuals $r_i = y_i - \boldsymbol{x}_i^\top \hat{\boldsymbol{\beta}}_{\text{WLS}}$. We then simply apply the exact same empirical shrinkage and boundary padding logic introduced previously directly to these residuals $r_i$ to establish the base grid. This ensures that the baseline component means perfectly capture the span of the unexplained residual heterogeneity before the covariate-driven location shifts are applied.

Let $\boldsymbol{X}$ be the $N \times p$ design matrix. We define the $N \times K$ matrix of scaled residuals, $\tilde{\boldsymbol{R}}$, with elements $\tilde{r}_{ik} = (y_i - (\mu_k + \boldsymbol{x}_i^\top \boldsymbol{\beta})) / (\tau_c^2 + v_i)$, and let $\tilde{\boldsymbol{r}}_i = (\tilde{r}_{i1}, \dots, \tilde{r}_{iK})^\top$ be the $K \times 1$ vector for study $i$. We then define $\boldsymbol{\bar{r}} = (\boldsymbol{Q} \circ \tilde{\boldsymbol{R}})\boldsymbol{1}_K$, where $\circ$ denotes the Hadamard (element-wise) product, and the $i$-th element is denoted as $\bar{r}_i$. Let $\phi_{ik} = \phi(y_i \mid \mu_k + \boldsymbol{x}_i^\top \boldsymbol{\beta}, \tau_c^2 + v_i)$ denote the study-specific component density. We define the combined parameter vector $\boldsymbol{\eta} = (\boldsymbol{\alpha}^\top, \boldsymbol{\beta}^\top)^\top$. The penalized log-likelihood function for the location-shifting model continues to apply the difference penalty to the shape parameters $\boldsymbol{\alpha}$:
\[l_p(\boldsymbol{\eta} \mid \lambda_\alpha) = l_p(\boldsymbol{\alpha}, \boldsymbol{\beta} \mid \lambda_\alpha) = \sum_{i=1}^N \log\left( \sum_{k=1}^K w_k \phi_{ik} \right) - \frac{\lambda_\alpha}{2} \boldsymbol{\alpha}^\top \boldsymbol{C}^\top \boldsymbol{D}^\top \boldsymbol{D} \boldsymbol{C} \boldsymbol{\alpha}.\]

\begin{theorem}[Penalized Score and Observed Information in Location-Shifting Model]
The penalized score vector partitions into:
\begin{align*}
\nabla_{\alpha} l_p &= \boldsymbol{C}^\top \left[ \boldsymbol{Q}^\top \boldsymbol{1}_N - N \boldsymbol{w} - \lambda_\alpha \boldsymbol{D}^\top \boldsymbol{D} \boldsymbol{C} \boldsymbol{\alpha} \right]\\
\nabla_\beta l_p &= \boldsymbol{X}^\top \boldsymbol{\bar{r}}
\end{align*}
Let $\boldsymbol{\tilde{V}} = \operatorname{diag}(\tau_c^2 + v_1, \dots, \tau_c^2 + v_N)$. The blocks of the penalized Observed Information Matrix are:
\begin{align*}
-\nabla^2_{\alpha\alpha} l_p &= \boldsymbol{C}^\top \left[ \boldsymbol{Q}^\top \boldsymbol{Q} - N \boldsymbol{w} \boldsymbol{w}^\top - \operatorname{diag}(\boldsymbol{Q}^\top \boldsymbol{1}_N - N \boldsymbol{w}) + \lambda_\alpha \boldsymbol{D}^\top \boldsymbol{D} \right] \boldsymbol{C}\\
-\nabla^2_{\beta\beta} l_p &= \boldsymbol{X}^\top \left[ \boldsymbol{\tilde{V}}^{-1} - \operatorname{diag}\big((\boldsymbol{Q} \circ \tilde{\boldsymbol{R}} \circ \tilde{\boldsymbol{R}})\boldsymbol{1}_K - \boldsymbol{\bar{r}} \circ \boldsymbol{\bar{r}} \big) \right] \boldsymbol{X}\\
-\nabla^2_{\alpha\beta} l_p &= \boldsymbol{C}^\top \left[ \boldsymbol{Q} \circ (\tilde{\boldsymbol{R}} - \boldsymbol{\bar{r}}\boldsymbol{1}_K^\top) \right]^\top \boldsymbol{X}
\end{align*}
\end{theorem}
\begin{proof}
The derivation of the score and the negative Hessian with respect to $\boldsymbol{\alpha}$ is similar to Theorem 1. The score with respect to $\boldsymbol{\beta}$ requires the derivative of the log-likelihood $l_i = \log \sum_k w_k \phi_{ik}$. We have $\frac{\partial l_i}{\partial \boldsymbol{\beta}} = \sum_k q_{ik} \frac{\partial \log \phi_{ik}}{\partial \boldsymbol{\beta}}$. Since $\phi_{ik}$ is a normal density with variance $\tilde{v}_i = \tau_c^2 + v_i$ and mean $\mu_{ik} = \mu_k + \boldsymbol{x}_i^\top \boldsymbol{\beta}$, $\frac{\partial \log \phi_{ik}}{\partial \boldsymbol{\beta}} = \frac{y_i - \mu_{ik}}{\tilde{v}_i} \boldsymbol{x}_i = \tilde{r}_{ik} \boldsymbol{x}_i$. Thus, $\frac{\partial l_i}{\partial \boldsymbol{\beta}} = \left( \sum_k q_{ik} \tilde{r}_{ik} \right) \boldsymbol{x}_i = \bar{r}_i \boldsymbol{x}_i$. Summing over $i$ yields $\boldsymbol{X}^\top \boldsymbol{\bar{r}}$.

$-\nabla^2_{\beta\beta} l_p$ requires the derivative of $\bar{r}_i$ with respect to $\boldsymbol{\beta}$: $\frac{\partial \bar{r}_i}{\partial \boldsymbol{\beta}} = \sum_k \left( \frac{\partial q_{ik}}{\partial \boldsymbol{\beta}} \tilde{r}_{ik} + q_{ik} \frac{\partial \tilde{r}_{ik}}{\partial \boldsymbol{\beta}} \right)$. Using $\frac{\partial \tilde{r}_{ik}}{\partial \boldsymbol{\beta}} = -\frac{1}{\tilde{v}_i} \boldsymbol{x}_i$ and $\frac{\partial q_{ik}}{\partial \boldsymbol{\beta}} = q_{ik}(\tilde{r}_{ik} - \bar{r}_i)\boldsymbol{x}_i$, we obtain $\frac{\partial \bar{r}_i}{\partial \boldsymbol{\beta}} = \left[ \sum_k q_{ik} \tilde{r}_{ik}^2 - \bar{r}_i^2 - \frac{1}{\tilde{v}_i} \right] \boldsymbol{x}_i$. Summing over $i$ in matrix form yields the specified Hessian block. The cross-derivative block $-\nabla^2_{\alpha\beta} l_p$ follows from taking the derivative of the $\boldsymbol{\alpha}^*$ score, $\boldsymbol{q}_i - \boldsymbol{w}$, with respect to $\boldsymbol{\beta}$, yielding $\operatorname{diag}(\boldsymbol{q}_i)(\tilde{\boldsymbol{r}}_i - \bar{r}_i\boldsymbol{1}_K)\boldsymbol{x}_i^\top$. Summation over $i$ and applying $\boldsymbol{C}$ produces the final matrix.
\end{proof}

Estimation proceeds via the Newton-Raphson method, utilizing the score and Hessian defined above. AIC is again used to select the optimal model. Once the penalized maximum likelihood estimates are obtained, quantities of interest can be derived in a straightforward manner. Because the location-shifting model strictly preserves the shape of the baseline distribution, any specific distributional quantities for a given covariate value---such as the mean, variance, PDF, CDF, tail probabilities, or quantiles---can be computed by simply shifting the baseline component means ($\mu_k$) by the estimated linear predictor ($\boldsymbol{x}^\top \hat{\boldsymbol{\beta}}$) and applying the exact same logic introduced in the intercept-only model.

The estimated location shifts are given directly by the fixed-effects vector $\hat{\boldsymbol{\beta}}$. The interpretation of $\hat{\boldsymbol{\beta}}$ in this location-shifting framework is directly analogous to coefficients in a standard meta-regression model. A coefficient $\hat{\beta}_j$ represents the estimated change in the overall mean of the true effects for a one-unit increase in the covariate $X_j$, holding all else constant. Because this model enforces a fixed distribution shape, a shift in the mean mathematically equates to a uniform shift of the entire probability density curve along the x-axis.

The decision to use this location-shifting model should be driven by both theoretical and empirical considerations. It is most appropriate when one can reasonably assume that after adjusting for the covariates, the residual heterogeneity maintains a similar shape and comparable magnitude across different contexts. Alternatively, it serves well when the expected contribution of the fixed-effect covariates is comparatively smaller than the vast, overarching residual heterogeneity. In such scenarios, adopting the location-shifting assumption provides a parsimonious summary without oversimplifying the complex reality of the underlying evidence base.

\subsubsection{Shape-Morphing Meta-Regression Model}
\label{sec:shape_meta_regression}

While the location-shifting approach provides a parsimonious summary by assuming a constant distributional shape across contexts, some covariates may fundamentally alter the structural dispersion of the effects rather than merely shifting their mean. To address this, the alternative shape-morphing model allows covariates to directly influence the mixture weights, causing the entire shape of the true effects distribution to mutate depending on study characteristics. Here we present this framework considering a single covariate $\boldsymbol{z}$. The shape-morphing model is defined as:
\begin{align*}
y_i \mid \theta_i &\sim \mathcal{N}(\theta_i, v_i)\\
\theta_i &\sim \sum_{k=1}^K w_{ik} \mathcal{N}(\mu_k, \tau_c^2)
\end{align*}
Crucially, the base grid of component means ($\mu_k$) is shared universally across all levels of the covariate. Because the grid locations remain fixed, the covariate $z_i$ modifies the mixture weights via a softmax transformation:
\[w_{ik} = \frac{\exp(\alpha^*_k + z_i \gamma^*_k)}{\sum_{l=1}^K \exp(\alpha^*_l + z_i \gamma^*_l)}\]

Figure~\ref{fig:shape_morphing} provides a visual demonstration of this dynamic. Unlike the location-shifting model, the baseline distribution (solid red line, $z=0$) fundamentally mutates across different covariate levels. As $z$ increases to $1$ (dashed green line), the model assigns higher weights to the right-side components, creating a heavier skew and elongating the tail. Conversely, when $z$ decreases to $-1$ (dotted blue line), the skew is suppressed, and the distribution concentrates more probability mass around its central peak. 

\begin{figure}[htbp]
    \centering
    \includegraphics[width=0.8\textwidth]{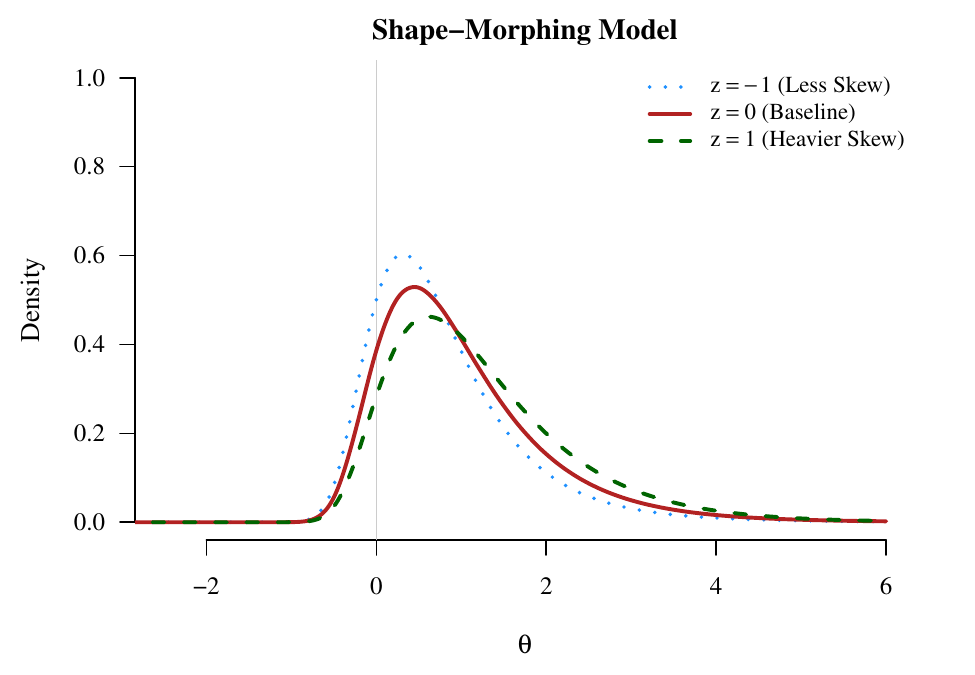}
    \caption{Illustration of the Shape-Morphing Model. The covariate $z$ directly influences the mixture weights, allowing the distribution to exhibit different structural properties---such as varying degrees of skewness---across different conditions.}
    \label{fig:shape_morphing}
\end{figure}

The shape modifications are captured by the parameters $\boldsymbol{\gamma}$, utilizing the same constraint matrix as the intercepts ($\boldsymbol{\gamma}^* = \boldsymbol{C}\boldsymbol{\gamma}$). Let $\boldsymbol{W}$ denote the $N \times K$ matrix of study-specific weights, with elements $w_{ik}$, and let $\boldsymbol{w}_i = (w_{i1}, \dots, w_{iK})^\top$ denote the $K \times 1$ weight vector for study $i$. We define the combined parameter vector $\boldsymbol{\eta} = (\boldsymbol{\alpha}^\top, \boldsymbol{\gamma}^\top)^\top$. The penalized log-likelihood applies penalties for both $\boldsymbol{\alpha}$ and $\boldsymbol{\gamma}$:
\[l_p(\boldsymbol{\alpha}, \boldsymbol{\gamma} \mid \lambda_\alpha, \lambda_\gamma) = \sum_{i=1}^N \log\left( \sum_{k=1}^K w_{ik} \phi_{ik} \right) - \frac{\lambda_\alpha}{2} \boldsymbol{\alpha}^\top \boldsymbol{C}^\top \boldsymbol{D}^\top \boldsymbol{D} \boldsymbol{C} \boldsymbol{\alpha} - \frac{\lambda_\gamma}{2} \boldsymbol{\gamma}^\top \boldsymbol{C}^\top \boldsymbol{D}^\top \boldsymbol{D} \boldsymbol{C} \boldsymbol{\gamma}\]

Just as $\lambda_\alpha$ controls the baseline smoothness of the distribution, the penalty parameter $\lambda_\gamma$ governs the complexity of the shape modification. By employing a third-order difference penalty on $\boldsymbol{\gamma}$, we ensure that changes in the weights across the grid remain smooth. In the limiting scenario where the penalty $\lambda_\gamma \to \infty$, the third-order differences of $\boldsymbol{\gamma}$ are heavily restricted. Similar to how an infinite $\lambda_\alpha$ reduces the baseline to a normal distribution, heavily penalizing $\boldsymbol{\gamma}$ restricts the covariate's influence to inducing only smooth, lower-order transformations (such as simple location or variance shifts of a normal baseline). This robust regularization prevents the model from fitting overly erratic, artifactual shapes.

\begin{theorem}[Penalized Score and Observed Information in Shape-Morphing Model]
The penalized score vector partitions into:
\begin{align*}
\nabla_{\alpha} l_p &= \boldsymbol{C}^\top \left[ (\boldsymbol{Q} - \boldsymbol{W})^\top \boldsymbol{1}_N - \lambda_\alpha \boldsymbol{D}^\top \boldsymbol{D} \boldsymbol{C} \boldsymbol{\alpha} \right]\\
\nabla_{\gamma} l_p &= \boldsymbol{C}^\top \left[ (\boldsymbol{Q} - \boldsymbol{W})^\top \boldsymbol{z} - \lambda_\gamma \boldsymbol{D}^\top \boldsymbol{D} \boldsymbol{C} \boldsymbol{\gamma} \right]
\end{align*}

Let $\boldsymbol{Z} = \operatorname{diag}(\boldsymbol{z})$ and $\boldsymbol{Z}^2 = \operatorname{diag}(\boldsymbol{z} \circ \boldsymbol{z})$. The blocks of the penalized Observed Information Matrix are:
\begin{align*}
-\nabla^2_{\alpha\alpha} l_p &= \boldsymbol{C}^\top \left[ \boldsymbol{Q}^\top \boldsymbol{Q} - \boldsymbol{W}^\top \boldsymbol{W} - \operatorname{diag}((\boldsymbol{Q} - \boldsymbol{W})^\top \boldsymbol{1}_N) + \lambda_\alpha \boldsymbol{D}^\top \boldsymbol{D} \right] \boldsymbol{C}\\
-\nabla^2_{\gamma \gamma} l_p &= \boldsymbol{C}^\top \left[ \boldsymbol{Q}^\top \boldsymbol{Z}^2 \boldsymbol{Q} - \boldsymbol{W}^\top \boldsymbol{Z}^2 \boldsymbol{W} - \operatorname{diag}((\boldsymbol{Q} - \boldsymbol{W})^\top (\boldsymbol{z} \circ \boldsymbol{z})) + \lambda_\gamma \boldsymbol{D}^\top \boldsymbol{D} \right] \boldsymbol{C}\\
-\nabla^2_{\alpha\gamma} l_p &= \boldsymbol{C}^\top \left[ \boldsymbol{Q}^\top \boldsymbol{Z} \boldsymbol{Q} - \boldsymbol{W}^\top \boldsymbol{Z} \boldsymbol{W} - \operatorname{diag}((\boldsymbol{Q} - \boldsymbol{W})^\top \boldsymbol{z}) \right] \boldsymbol{C}
\end{align*}
\end{theorem}
\begin{proof}
Similar to Theorem 1, the derivative of the log-likelihood for study $i$ with respect to the linear predictor terms $\alpha_k^* + z_i \gamma_k^*$ is $q_{ik} - w_{ik}$. By the chain rule, the derivative with respect to $\alpha_k^*$ is $q_{ik} - w_{ik}$, and with respect to $\gamma_k^*$ it is $z_i(q_{ik} - w_{ik})$. Summing over all studies and applying the constraint matrix $\boldsymbol{C}$ and respective penalties yields the specified score vectors.

For the Hessian, we differentiate $\boldsymbol{q}_i - \boldsymbol{w}_i$ with respect to $\boldsymbol{\alpha}^*$ and $\boldsymbol{\gamma}^*$. The derivative with respect to $\boldsymbol{\alpha}^*$ is $\operatorname{diag}(\boldsymbol{q}_i - \boldsymbol{w}_i) - \boldsymbol{q}_i \boldsymbol{q}_i^\top + \boldsymbol{w}_i \boldsymbol{w}_i^\top$, which sums to the $\alpha\alpha$ block. Differentiating with respect to $\boldsymbol{\gamma}^*$ introduces the inner derivative $z_i$, resulting in $\operatorname{diag}(\boldsymbol{q}_i - \boldsymbol{w}_i)z_i^2 - (\boldsymbol{q}_i \boldsymbol{q}_i^\top - \boldsymbol{w}_i \boldsymbol{w}_i^\top)z_i^2$, yielding the $\gamma\gamma$ block containing $\boldsymbol{Z}^2$. The cross-derivative block $\alpha\gamma$ similarly introduces a single $z_i$ factor, resulting in the block containing $\boldsymbol{Z}$.
\end{proof}

Consistent with the preceding models, estimation relies on the Newton-Raphson algorithm. To accurately tune the model's complexity, we perform a two-dimensional grid search over combinations of the smoothing parameters $\lambda_\alpha$ and $\lambda_\gamma$ to find the optimal pair that minimizes the AIC. Once the penalized maximum likelihood estimates ($\hat{\boldsymbol{\alpha}}$ and $\hat{\boldsymbol{\gamma}}$) are obtained, we can derive conditional estimates for any specific target covariate value $z$. The conditional mixture weights are evaluated using the estimated linear predictor: $\hat{w}(z)_k = \frac{\exp(\hat{\alpha}^*_k + z \hat{\gamma}^*_k)}{\sum_{l=1}^K \exp(\hat{\alpha}^*_l + z \hat{\gamma}^*_l)}$.

With these conditional weights established, generating context-specific distributional profiles becomes straightforward. Any conditional distributional quantity is computed simply by substituting the conditional weights $\hat{\boldsymbol{w}}(z)$ directly into the base formulas introduced in the intercept-only model. For example, the conditional mean is estimated as $\hat{\mu}(z) = \boldsymbol{\mu}^\top \hat{\boldsymbol{w}}(z)$. More importantly, this framework allows researchers to evaluate quantities that depend on the full shape of the distribution. For instance, researchers can estimate the conditional probability of an adverse effect ($P(\theta < 0 \mid z)$), allowing them to answer questions such as: ``What is the probability of this treatment causing harm when administered to an older population ($z = \text{high}$)?'' Similarly, one could evaluate the conditional 10th percentile of effects ($q_{0.10} \mid z$) by finding the root that inverts the conditional CDF.

The power of this conditional estimation lies in its ability to explicitly evaluate how the entire shape of the true effects distribution mutates across different contexts. Rather than merely tracking mean differences, researchers can compute the change in any targeted quantity between two distinct covariate levels, $z_1$ and $z_2$. For example, by comparing the conditional variances ($\tau^2$), researchers can test if one setting produces more consistent results than another. Furthermore, by comparing tail probabilities, researchers can directly assess whether changing a study characteristic significantly reduces the risk of negative outcomes ($\Delta P(\theta < 0)$). This transforms meta-regression from a tool that merely tracks the center of the data into one that evaluates changes in the safety and reliability of interventions across contexts.

While the mathematical framework of the shape-morphing model can theoretically be extended to multiple covariates, doing so introduces significant practical hurdles. Introducing just one shape covariate immediately doubles the number of free mixture parameters (from $\boldsymbol{\alpha}$ alone to both $\boldsymbol{\alpha}$ and $\boldsymbol{\gamma}$). Adding more covariates linearly multiplies this vast parameter space, severely risking overfitting, especially given the sample sizes typical of meta-analyses. Furthermore, model selection becomes exponentially harder: tuning the model would require evaluating a higher-dimensional grid of penalty parameters (e.g., searching across $\lambda_\alpha, \lambda_{\gamma_1}, \lambda_{\gamma_2}, \dots$). Such an expansive search is practically unachievable and computationally prohibitive unless novel, more efficient model selection algorithms or sparsity-inducing simplifications are developed.

\subsection{Handling Dependent Effect Sizes}
\label{sec:handling_dependency}

In large social science meta-analyses, it is exceptionally common to encounter dependent effect sizes. This complexity arises from various nested or crossed data structures, such as when multiple correlated outcomes are extracted from the same sample, when several treatment arms share a common control group, when outcomes are measured at multiple follow-up time points, or when a single research laboratory publishes multiple related studies. Treating these grouped observations as fully independent violates fundamental statistical assumptions, which can lead to underestimated standard errors, spuriously narrow confidence intervals, and inflated Type I error rates.

To account for effect size dependence without sacrificing the distributional flexibility of our model, we adopt an approach that combines a working independence model with Cluster-Robust Variance Estimation (CRVE). In essence, we estimate the parameters of the Penalized Gaussian Mixture by maximizing a composite, or pseudo-likelihood~\citep{Varin2011}. During the optimization phase, we treat the effect sizes as if they were independent, and we subsequently correct the standard errors post hoc to account for the correlation.

One might naturally ask why we do not specify a more nuanced working model that explicitly models the complex dependence structure. In traditional normal-normal meta-analysis, researchers often use multivariate extensions (e.g., multilevel or multivariate normal models) to explicitly partition heterogeneity into within-cluster and between-cluster components. However, attempting to extend the PGM framework to a multivariate setting is practically infeasible. Building flexible mixtures in higher-dimensional spaces introduces severe computational limitations due to the curse of dimensionality, requiring an exponentially larger grid of basis functions. Furthermore, explicitly modeling the joint distribution requires strong, often unjustified assumptions about the true Data Generating Process (DGP) of the dependence structure, which is rarely known in complex meta-analytic datasets.

By intentionally choosing a pseudo-likelihood approach based on a marginal working model, we accept a specific methodological trade-off. The primary downside is that we are limited to estimating the flexible marginal distribution of the true effects; we cannot partition the heterogeneity into distinct hierarchical levels, for example, between-study and within-study between-effect-size heterogeneity. However, in the context of summarizing highly heterogeneous interventions, researchers are predominantly interested in characterizing the overall, marginal distribution of the entire population of effects anyway. For this primary objective, the working independence model coupled with CRVE serves exceptionally well. It sidesteps intractable computational costs, avoids unjustified assumptions about the joint DGP, and provides valid inference for the marginal shape of the distribution.

Formally, suppose the $N$ total effect sizes are nested within $M$ independent clusters (e.g., independent studies or independent samples). Let $C_m$ denote the set of observation indices belonging to cluster $m$, where $m = 1, \dots, M$. Let $\boldsymbol{\eta}$ denote the full vector of parameters (e.g., $[\boldsymbol{\alpha}^\top, \boldsymbol{\beta}^\top]^\top$ or $[\boldsymbol{\alpha}^\top, \boldsymbol{\gamma}^\top]^\top$). Recall that under strict independence, the covariance matrix is estimated via $\boldsymbol{V}_\eta = \boldsymbol{J}_p^{-1} \boldsymbol{J} \boldsymbol{J}_p^{-1}$. Under dependency, we extend this to a robust empirical sandwich estimator.

For the $i$-th effect, let $\boldsymbol{s}_i$ be its specific vector contribution to the unpenalized score (gradient). We aggregate these score vectors within each cluster to form a cluster-level score vector $\boldsymbol{s}_m^{(c)}$:
\[\boldsymbol{s}_m^{(c)} = \sum_{i \in C_m} \boldsymbol{s}_i\]

The cluster-robust ``meat'' matrix, $\boldsymbol{J}_{CR}$, is then constructed as the sum of the outer products of these aggregated cluster-level scores across all $M$ clusters:
\[\boldsymbol{J}_{CR} = \sum_{m=1}^M \boldsymbol{s}_m^{(c)} {\boldsymbol{s}_m^{(c)}}^\top\]

Finally, the robust variance-covariance matrix for the parameter estimates replaces the unpenalized meat matrix $\boldsymbol{J}$ with the cluster-robust version $\boldsymbol{J}_{CR}$:
\[\boldsymbol{V}_{CR} = \boldsymbol{J}_p^{-1} \boldsymbol{J}_{CR} \boldsymbol{J}_p^{-1}\]

By utilizing $\boldsymbol{J}_{CR}$ in the center of the sandwich estimator, this adjustment ensures asymptotically valid inference under an unknown within-cluster correlation structure. The resulting robust covariance matrix $\boldsymbol{V}_{CR}$ can then be substituted directly into the Delta method derivations to yield cluster-robust confidence intervals for all target quantities, including the mean, variance, and tail probabilities.

\section{Simulation Study}

\subsection{Simulation 1: Independent Effect Sizes}

\subsubsection{Simulation Setup}

To evaluate the performance of the proposed Penalized Gaussian Mixture (PGM) method relative to the standard random-effects meta-analysis (normal-normal model), we conducted two sets of simulations. In this first simulation study, we focus on the foundational case of strictly independent effect sizes. The primary objective is to assess how well each method recovers the true underlying distribution of true effects under varying conditions of sample size, $I^2$, and distributional shape, establishing a baseline before introducing complex dependent data structures.

We generated data mimicking the structure of typical meta-analyses. The true effect sizes $\theta_i$ were generated from one of three distinct distributions to represent different underlying realities. Notably, these chosen distributions are exactly the same three distributions depicted in Figure \ref{fig:distributions}. Crucially, we fixed the overall mean and variance for all three distributions to be exactly 1 ($\mu=1$, $\tau^2=1$). The first is a Normal distribution, $\theta_i \sim \mathcal{N}(1, 1)$, serving as a baseline where standard conventional assumptions hold true. The second is a shifted Log-Normal distribution, representing a scenario with significant right-skewness. Specifically, we simulated an initial variable from a log-normal distribution with log-mean $\mu_{\log} = 0$ and log-sd $\sigma_{\log} = 0.5$, and then shifted and scaled it to ensure the final distribution has a mean of 1 and a variance of 1. The third is a bimodal Mixture of two normal distributions, creating a shape with a dominant positive cluster and a smaller negative cluster. This was parameterized as $0.3\mathcal{N}(\mu_1, \sigma_{\text{mix}}^2) + 0.7\mathcal{N}(\mu_2, \sigma_{\text{mix}}^2)$, where $\mu_1 = 1 - 2.1/\sqrt{2.14}$, $\mu_2 = 1 + 0.9/\sqrt{2.14}$, and $\sigma_{\text{mix}} = 0.5/\sqrt{2.14}$.

We simulated $v_i$ directly from a Gamma distribution. In a random-effects meta-analysis, the $I^2$ statistic can be approximated by $I^2 = \tau^2 / (\tau^2 + \bar{v})$, where $\bar{v}$ is the harmonic mean of the within-study variances, given by $1/E[1/v_i]$. Because we fixed the between-study variance at $\tau^2 = 1$, we can manipulate $I^2$ strictly by changing the distribution of $v_i$. To achieve $I^2$ values of 80\%, 50\%, and 20\%, we require $\bar{v}$ target values of 0.25, 1, and 4, respectively. We accomplished this by drawing $v_i \sim \text{Gamma}(\text{shape} = 10, \text{scale} = \theta)$, where the scale parameter $\theta \in \{0.25/9, 1/9, 4/9\}$. Because the harmonic mean of this specific Gamma distribution simplifies directly to $(\text{shape}-1)\theta = 9\theta$, this process exactly yields the desired $\bar{v}$ target values.

We varied the number of studies $N \in \{25, 50, 100, 200, 400\}$ to cover the typical range from small to very large meta-analyses in the social sciences. This resulted in a total of $5 \times 3 \times 3 = 45$ distinct simulation scenarios. For each scenario, we performed 10,000 replications.

We implemented the standard random-effects model using Maximum Likelihood (ML) estimation via the \texttt{metafor} package~\citep{Viechtbauer2010} in R as the conventional benchmark. We specifically selected the ML estimator for the normal model to ensure a fair comparison, as our proposed PGM framework is fundamentally a penalized likelihood method. For the proposed framework, we developed an open-source R package, \texttt{metapgm}, which can be easily installed via the command \texttt{pak::pkg\_install("daihesui/metapgm")}. The PGM model was fitted using the \texttt{rma.pgm} function from this package. To maximize the penalized likelihood, our package utilizes a trust-region optimization algorithm via the \texttt{trust} package~\citep{Geyer2026}, which offers superior numerical stability and convergence reliability compared to standard Newton-Raphson implementations.

For the PGM specifications, we dynamically selected the mixture complexity to scale with the available data, setting the number of components to $K = \text{round}(10 \log_{10}(N))$ to provide an appropriately flexible basis. We enforced smoothness using a third-order difference penalty ($d=3$). The optimal smoothing parameter $\lambda_\alpha$ was selected by minimizing the Akaike Information Criterion (AIC) across a predefined grid of values on a logarithmic scale: $e^{-5}, e^{-4}, \dots, e^5$.

To comprehensively assess parameter accuracy and inference quality, we evaluated seven specific targets for which the true values are exactly known. First, we calculated the Root Mean Square Error (RMSE) and the empirical coverage rate of the 95\% confidence intervals for the estimates of the overall mean ($\mu$) and the between-study variance ($\tau^2$); their true values are both exactly 1. Furthermore, to evaluate the models' practical utility in risk assessment and capturing the full distribution shape, we measured their accuracy and coverage in estimating tail probabilities. For each distribution, we established thresholds at the true 10th, 25th, 50th, 75th, and 90th percentiles ($q_{0.10}, q_{0.25}, q_{0.50}, q_{0.75}, q_{0.90}$). We then estimated the cumulative probability of effects falling below these thresholds, $P(\theta < q_p)$. Because these thresholds correspond to exact underlying quantiles, the target probabilities are precisely 0.10, 0.25, 0.50, 0.75, and 0.90. We computed the RMSE and the 95\% confidence interval coverage for these five distinct probability estimates. For the conventional normal model, the confidence intervals for these tail probabilities were constructed utilizing the methods described by~\citet{Mathur2019}.

Finally, to measure the global accuracy of the estimated density shape, we computed the Median Integrated Absolute Error (MIAE) between the estimated density $\hat{f}(\theta)$ and the true density $f(\theta)$. For each replication, we first calculated the Integrated Absolute Error (IAE):
\[\text{IAE} = \int |\hat{f}(\theta) - f(\theta)| d\theta.\]
The MIAE for a given scenario was then obtained by taking the median of these IAE values across all 10,000 replications. All simulation code and the full results are available at \url{https://osf.io/2hzme}.

\subsubsection{Simulation Results}

The simulation study evaluates the performance of the Penalized Gaussian Mixture (PGM) framework against the standard normal-normal approach (Normal-ML) across varying sample sizes, $I^2$, and true underlying distributions. The results demonstrate the distinct advantages of the PGM model in non-normal scenarios, highlighting critical failures in standard inference when underlying assumptions are violated.

We first examine the basic summary statistics: the overall mean ($\mu$) and the between-study variance ($\tau^2$). A common concern with flexible, over-parameterized models is that they might lose accuracy when estimating these fundamental parameters. However, the results indicate that the proposed PGM method does not sacrifice point estimation accuracy in exchange for its flexibility. Figure \ref{fig:rmse_mean} displays the Root Mean Square Error (RMSE) of the mean estimates, while Figure \ref{fig:rmse_var} illustrates the RMSE of the variance. Across all conditions---Normal, Lognormal, and Mixture distributions---the PGM and Normal-ML models yield virtually identical error rates. For both methods, the errors consistently decay as the sample size grows from 25 to 400, indicating that the flexible PGM model successfully recovers the true underlying moments.

\begin{figure}[htbp]
    \centering
    \includegraphics[width=\textwidth]{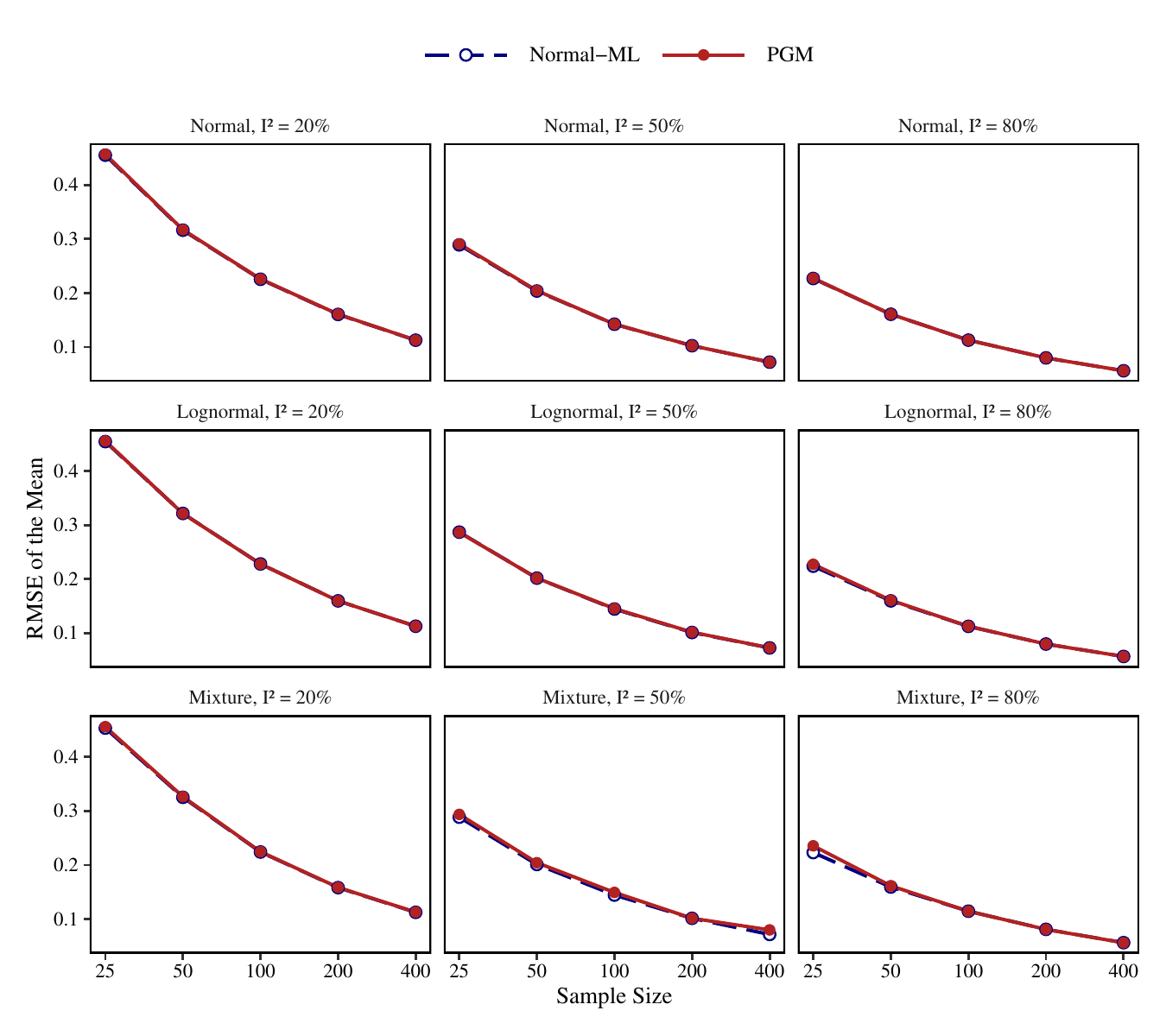}
    \caption{Root Mean Square Error (RMSE) of the estimated overall mean across varying true distributions, $I^2$, and sample sizes.}
    \label{fig:rmse_mean}
\end{figure}

\begin{figure}[htbp]
    \centering
    \includegraphics[width=\textwidth]{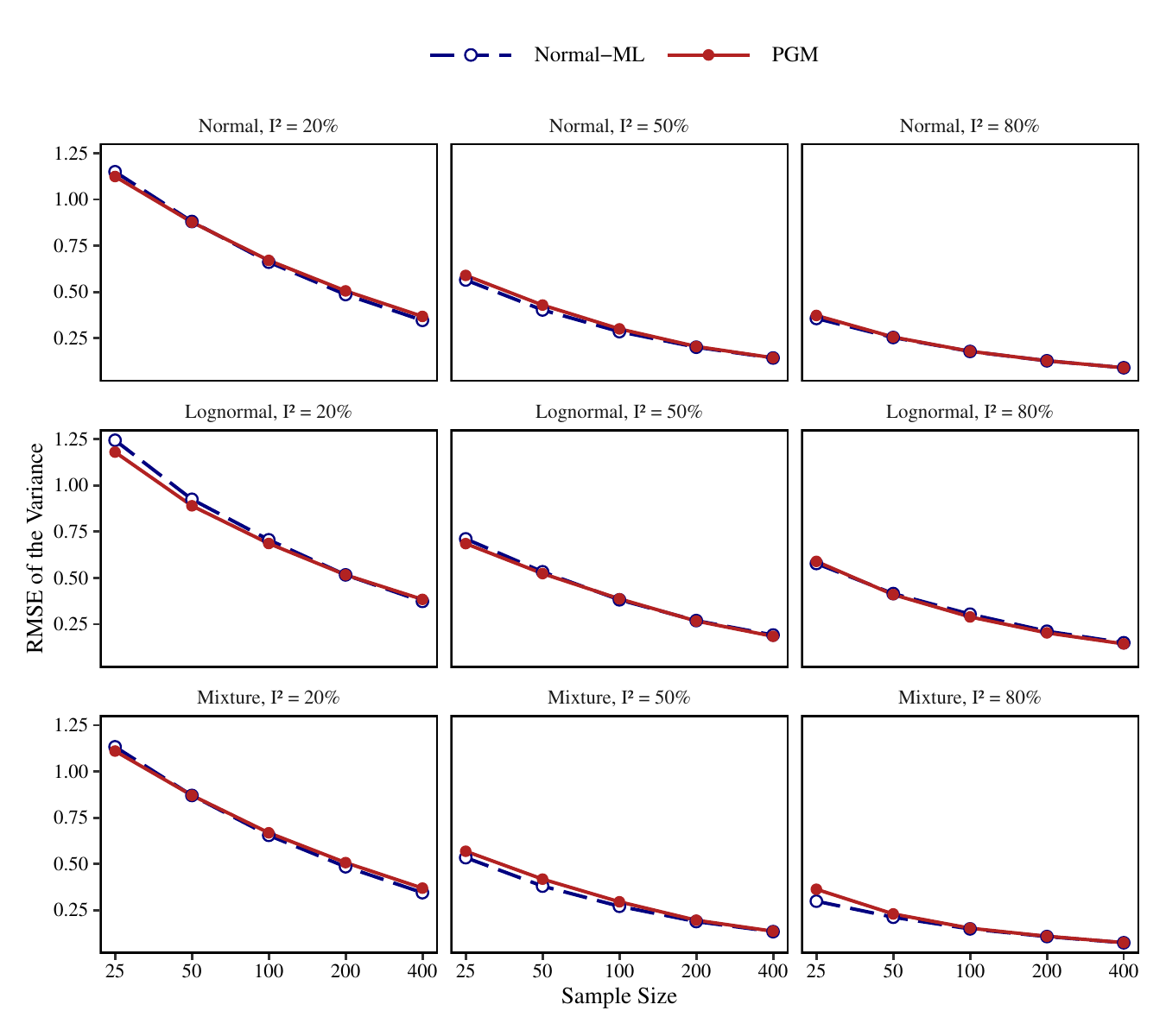}
    \caption{Root Mean Square Error (RMSE) of the estimated between-study variance ($\tau^2$) across varying true distributions, $I^2$, and sample sizes.}
    \label{fig:rmse_var}
\end{figure}

While the point estimates are highly comparable, differences emerge when evaluating statistical inference at finite sample sizes. Figure \ref{fig:cov_mean} presents the empirical coverage rates for the 95\% confidence intervals of the mean. The Normal-ML model maintains near-nominal coverage across most scenarios. The PGM model exhibits slight undercoverage at smaller sample sizes ($N=25$ and $50$), but successfully converges toward the nominal 0.95 level as the sample size increases to 100 and beyond. 

\begin{figure}[htbp]
    \centering
    \includegraphics[width=\textwidth]{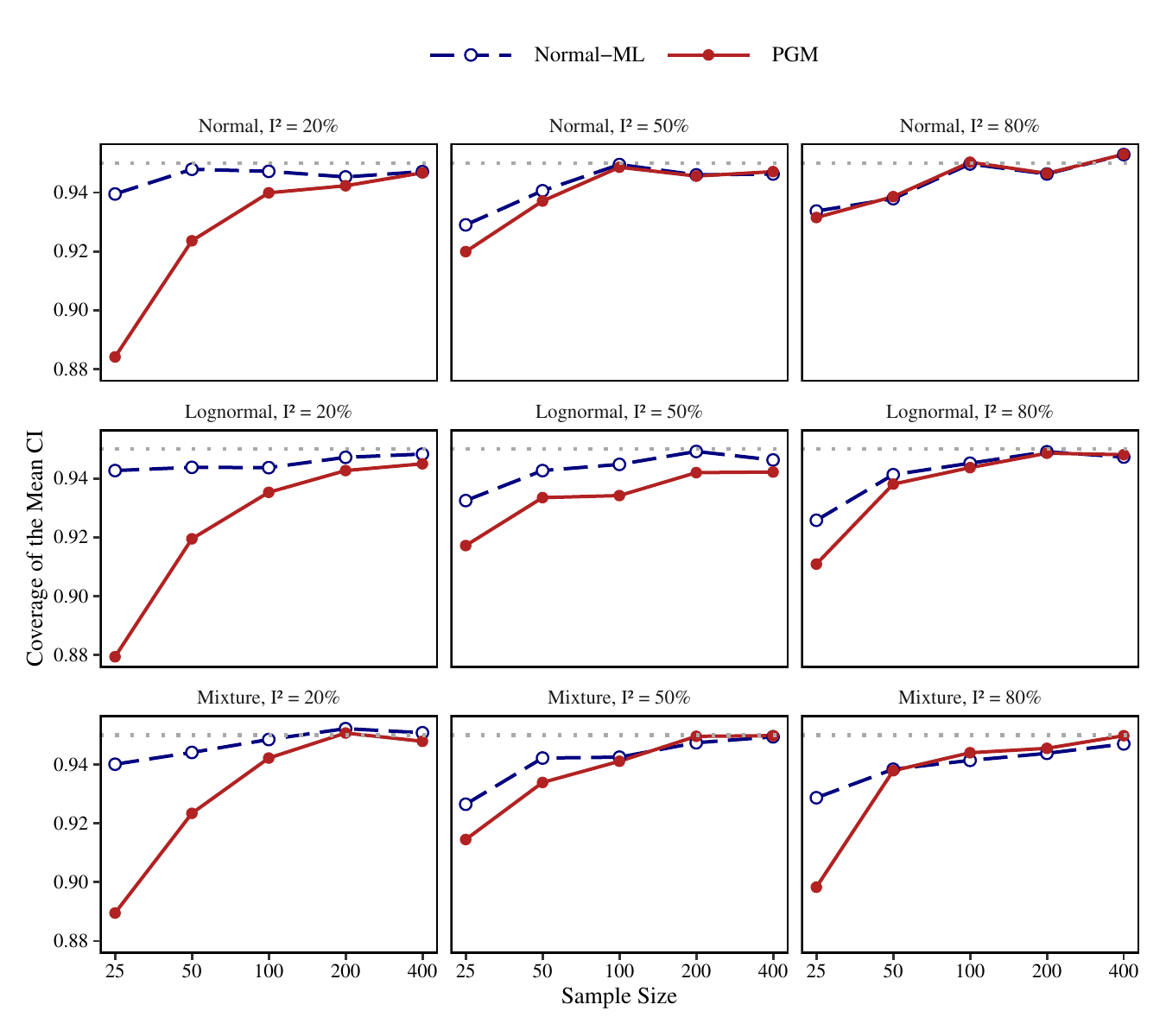}
    \caption{Empirical coverage rates of the 95\% confidence intervals for the overall mean across varying true distributions, $I^2$, and sample sizes.}
    \label{fig:cov_mean}
\end{figure}

This finite-sample problem is more pronounced for the variance parameter, and the behavior varies significantly depending on the true underlying distribution, as shown in Figure \ref{fig:cov_var}. The confidence intervals of the variance from the PGM models suffer from notable undercoverage at small sample sizes. However, this coverage steadily climbs toward nominal levels as the sample size reaches 200 and 400. When the true effects are Normal, both the PGM and Normal-ML models converge to the nominal 0.95 level as the sample size increases. However, a stark divergence occurs in the non-normal scenarios. Under the heavily skewed Lognormal distribution, the Normal-ML model suffers from structural failure; its coverage flatlines---stagnating around 0.75 for high $I^2$ conditions---and completely fails to improve even at $N=400$, likely because the fundamental parametric assumption is violated. In contrast, while the PGM model experiences severe undercoverage at small sample sizes in this scenario, its coverage is responsive to data, systematically climbing toward the nominal level as $N$ grows. A similar finite-sample trajectory for the PGM model is observed in the bimodal Mixture distribution, where it reliably recovers to the 0.95 level at larger sample sizes, whereas the Normal-ML model sometimes yields artificially high, conservative coverage. This divergence highlights a trade-off: capturing complex, asymmetric, or bimodal distributional shapes requires estimating a larger number of parameters, which inherently demands more data to achieve reliable inference on second-order moments. Yet, unlike standard methods that remain permanently biased under misspecification, the PGM framework consistently converges toward accurate inference as the sample size grows.

\begin{figure}[htbp]
    \centering
    \includegraphics[width=\textwidth]{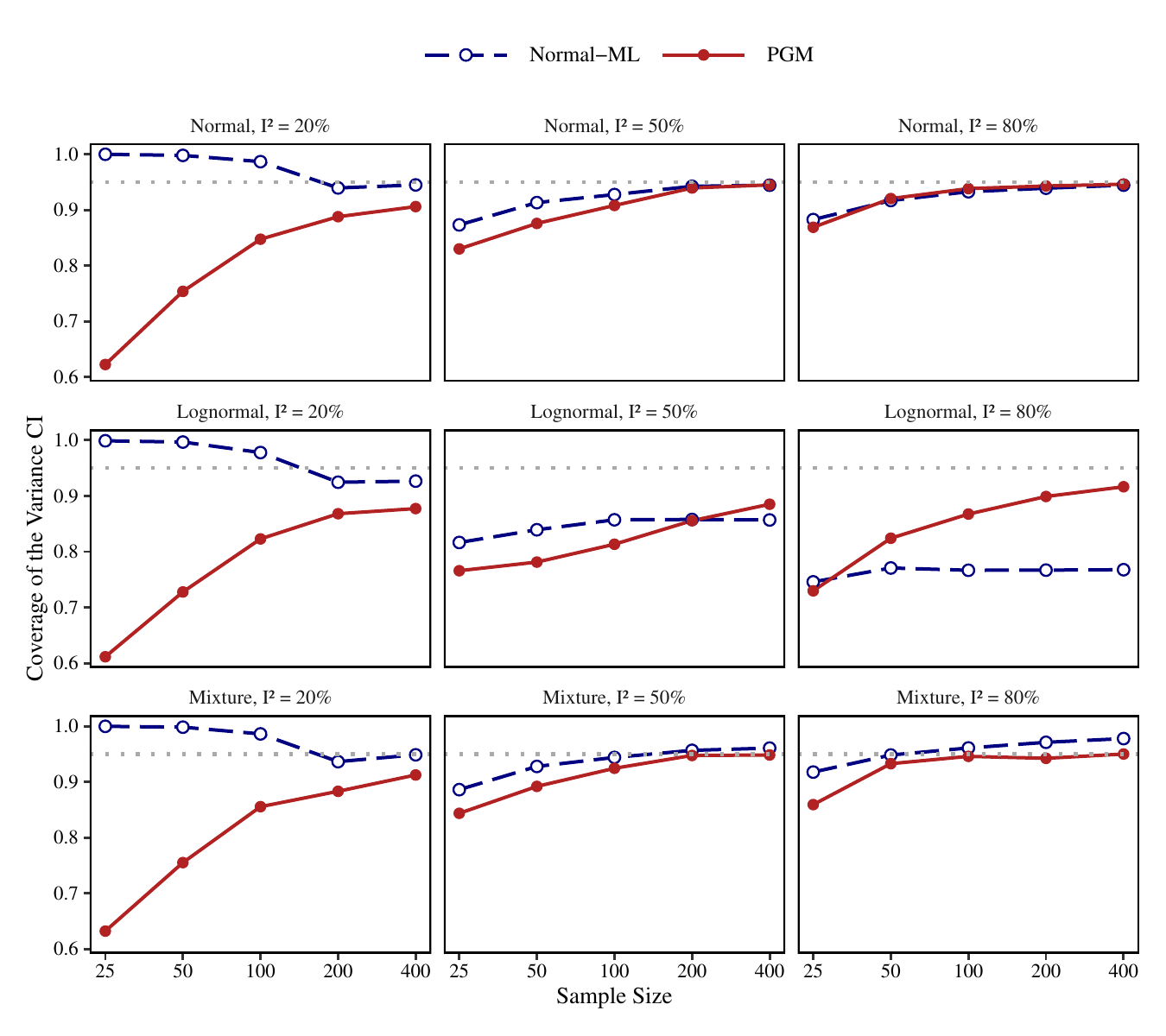}
    \caption{Empirical coverage rates of the 95\% confidence intervals for the between-study variance ($\tau^2$) across varying true distributions, $I^2$, and sample sizes.}
    \label{fig:cov_var}
\end{figure}

Although the standard normal model performs adequately for basic moments, the primary advantage of the PGM framework becomes undeniable when estimating distributional properties beyond the mean and variance. Figure \ref{fig:rmse_prob} presents the Average RMSE of Tail Probabilities, computed across the true 10th, 25th, 50th, 75th, and 90th percentiles. Under the true Normal distribution, the Normal-ML model naturally serves as the optimal baseline, slightly outperforming PGM. However, when the true effects follow a Lognormal or Mixture distribution, the rigid symmetric assumptions of the Normal-ML model cause its RMSE to plateau; accumulating more studies fails to correct the model's fundamental misspecification. The PGM model dynamically adapts to the data, consistently refining its tail probability estimates and yielding significantly lower RMSE as $N$ increases.

\begin{figure}[htbp]
    \centering
    \includegraphics[width=\textwidth]{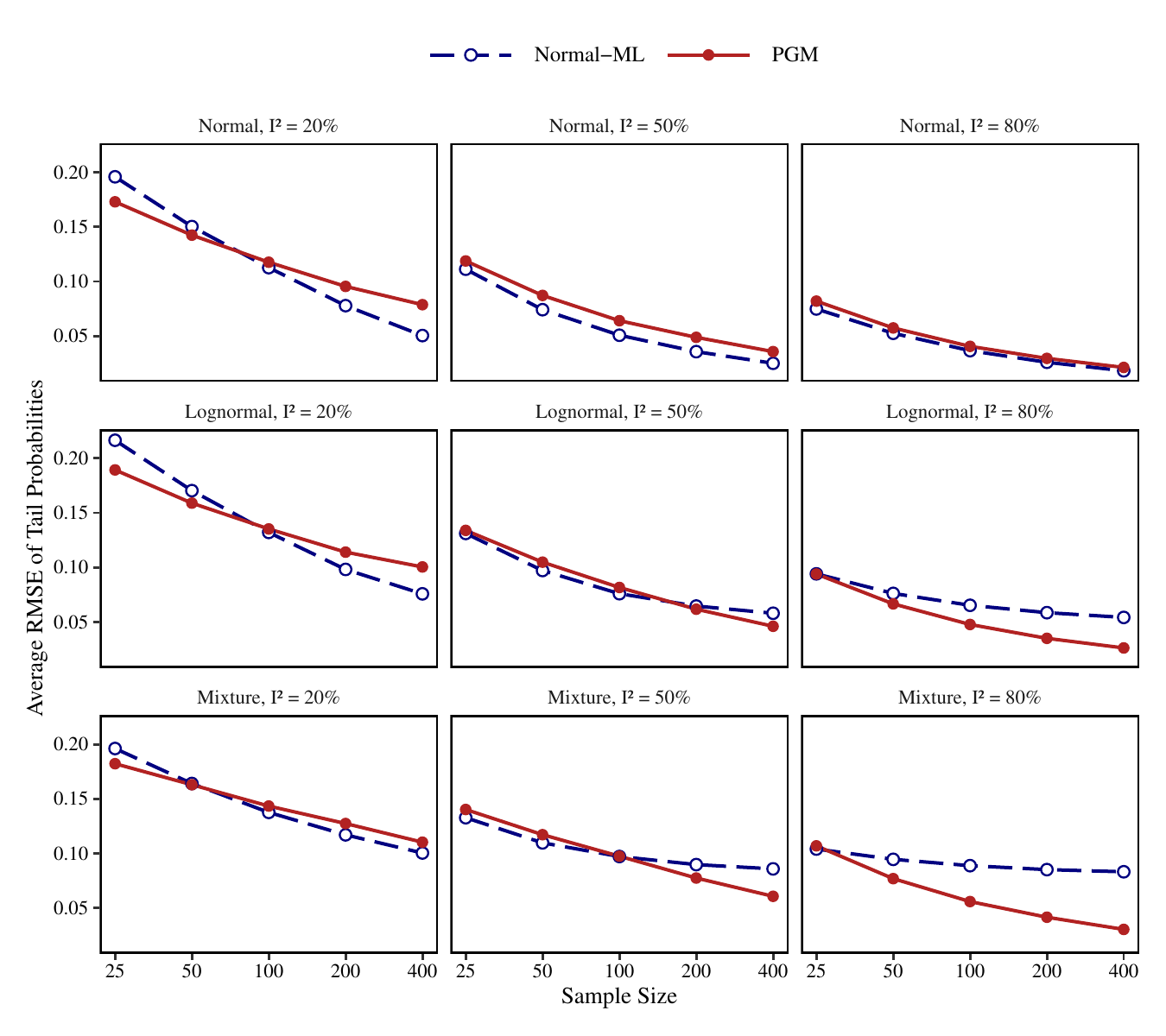}
    \caption{Average Root Mean Square Error (RMSE) of the estimated tail probabilities across varying true distributions, $I^2$, and sample sizes.}
    \label{fig:rmse_prob}
\end{figure}

The consequences of this misspecification are most drastically observed in the confidence interval coverage for these tail probabilities, shown in Figure \ref{fig:cov_prob}. In the Lognormal and Mixture scenarios, the coverage of the standard Normal-ML method completely collapses as the sample size increases, plunging as low as 0.2 at $N=400$. This occurs because the standard model becomes highly confident in an entirely incorrect, symmetric shape. In stark contrast, the PGM framework maintains robust, stable coverage near the nominal 0.95 level across these complex realities, proving essential for accurate risk assessment and subgroup detection.

\begin{figure}[htbp]
    \centering
    \includegraphics[width=\textwidth]{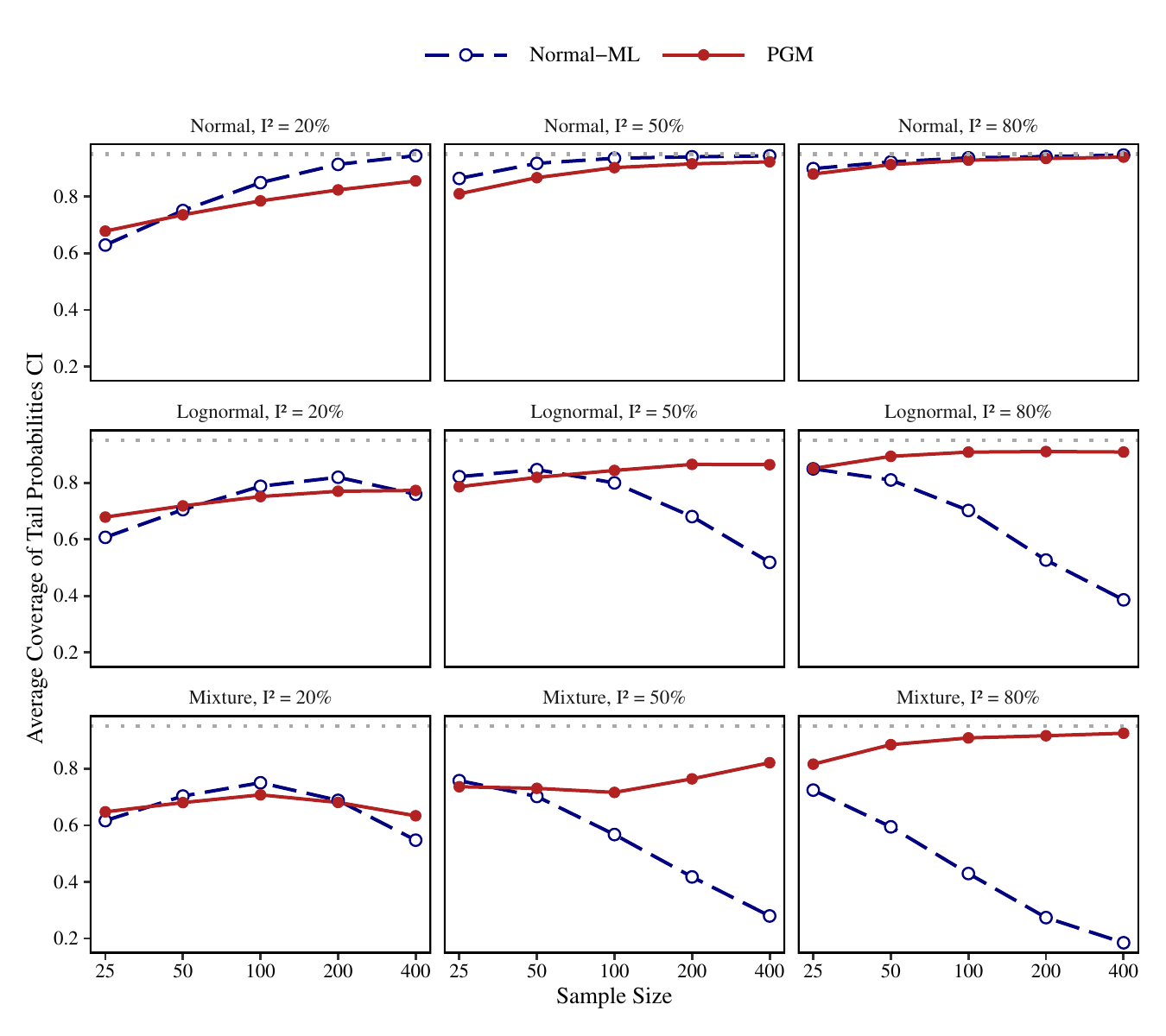}
    \caption{Average empirical coverage rates of the 95\% confidence intervals for the tail probabilities across varying true distributions, $I^2$, and sample sizes. Note the catastrophic collapse in coverage for the Normal-ML method in non-normal scenarios as sample size increases.}
    \label{fig:cov_prob}
\end{figure}

Finally, to quantify how well the models capture the holistic shape of the true effects, Figure \ref{fig:miae} plots the Median Integrated Absolute Error (MIAE) of the estimated probability density functions. As expected, when the underlying data is truly Normal, the parametric Normal-ML method achieves the lowest MIAE, though the PGM model adequately approaches this performance bound at larger sample sizes. However, the true cost of incorrectly assuming normality is starkly apparent in the non-normal conditions. For the asymmetric and bimodal distributions, the Normal-ML baseline exhibits severe bias, resulting in high MIAE values that remain elevated regardless of sample size. By utilizing its weighted Gaussian components, the PGM method successfully approximates the heavy skewness and distinct sub-populations, resulting in drastically lower MIAE scores that systematically decrease toward zero as the sample size grows.

\begin{figure}[htbp]
    \centering
    \includegraphics[width=\textwidth]{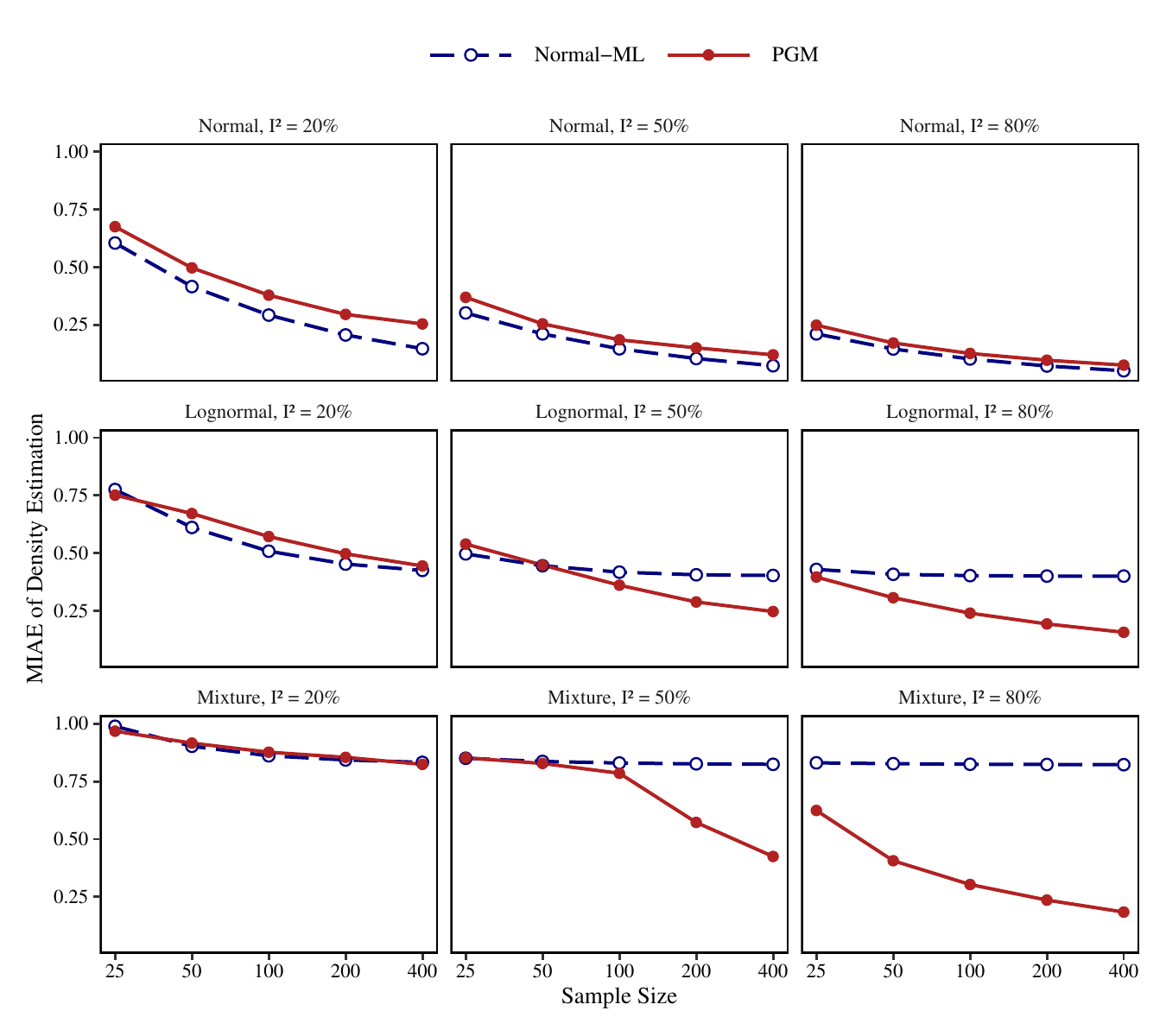}
    \caption{Median Integrated Absolute Error (MIAE) of the estimated probability density functions across varying true distributions, $I^2$, and sample sizes.}
    \label{fig:miae}
\end{figure}

\subsection{Simulation 2: Dependent Effect Sizes}

\subsubsection{Simulation Setup}

In social science meta-analyses, it is exceptionally common for multiple effect sizes to be extracted from the same primary study or cluster, creating a nested data structure. Our second simulation study aims to evaluate the performance of the proposed PGM framework in the presence of these dependent effect sizes. Specifically, we investigate whether combining a working independence PGM model with Cluster-Robust Variance Estimation (CRVE) as detailed in Section \ref{sec:handling_dependency} yields valid statistical inference, and we compare this approach against a naive PGM model that entirely ignores the clustering structure.

To simulate this dependency, we used the following data-generating process. Let $y_{ij}$ denote the $j$-th observed effect size in the $i$-th cluster. The data were generated according to the following model:
\[y_{ij} = \beta x_{i} + \sqrt{r} s_i + \sqrt{1-r} u_{ij} + e_{ij}\]
where the fixed-effect slope is set to $\beta = 1$, and the covariate $x_{i}$ is drawn from a standard normal distribution, $\mathcal{N}(0, 1)$.

The true effect heterogeneity is partitioned into two components: a cluster-level random effect $s_i$ and an effect-level random effect $u_{ij}$. To ensure a complex, non-normal reality, the cluster-level effect $s_i$ was generated from the same bimodal mixture distribution used in Simulation 1, but shifted to have a mean of 0 and a variance of 1. The effect-level effect $u_{ij}$ was generated from a standard normal distribution, $\mathcal{N}(0, 1)$. The parameter $r$ controls the proportion of the total true heterogeneity attributable to the cluster level, analogous to the intra-class correlation. We varied $r \in \{0, 0.4, 0.8\}$, representing scenarios ranging from perfect within-cluster independence to high within-cluster correlation. The sampling error $e_{ij}$ was drawn from $\mathcal{N}(0, v_{ij})$, where the within-study variances $v_{ij}$ were generated from the same Gamma distributions used in Simulation 1 to achieve target harmonic means, effectively producing approximate $I^2$ values of 20\%, 50\%, and 80\%.

We varied the number of independent clusters $M \in \{10, 20, 40, 80, 160\}$. To generate realistic cluster sizes, the number of effect sizes within each cluster was drawn using a shifted Poisson distribution. Specifically, we drew sizes from a Poisson distribution with a rate parameter of 3 and added 1 to each draw to make sure each cluster contains at least one effect size. This factorial design resulted in $5 \times 3 \times 3 = 45$ distinct scenarios. We performed 10,000 replications for each scenario.

For each generated dataset, we fit two meta-regression models to adjust for the covariate $x_{i}$. The first is the standard Penalized Gaussian Mixture (PGM) model, which naively assumes all effect sizes are strictly independent. The second, denoted as PGM-CRVE, is the proposed approach detailed in Section \ref{sec:handling_dependency} that fits the PGM using a working independence model but applies Cluster-Robust Variance Estimation post hoc to adjust the standard errors for the $M$ clusters.

Given that both the naive PGM and the PGM-CRVE models yield identical point estimates, our comparative evaluation focuses on the accuracy of their statistical inference. Accordingly, we evaluated the performance of the two models using two specific coverage metrics. First, we assessed the empirical coverage rate of the 95\% confidence interval for the fixed-effect slope ($\beta$). Second, we evaluated the empirical coverage for conditional tail probabilities. Similar to Simulation 1, we defined target thresholds at the true 10th, 25th, 50th, 75th, and 90th percentiles of the true marginal effects distribution evaluated at $x=0$. Because the true heterogeneity at $x=0$ is a convolution of a bimodal mixture ($s_i$) and a normal distribution ($u_{ij}$), we computed these exact true quantiles $q_p$ for each value of $r$ prior to evaluation. We then calculated the empirical coverage of the 95\% confidence intervals for these five conditional probabilities, $P(\theta < q_p \mid x=0)$. All simulation code and the full results are available at \url{https://osf.io/2hzme}.

\subsubsection{Simulation Results}

The second simulation study evaluates the effectiveness of the proposed cluster-robust variance estimation (CRVE) approach in handling dependent effect sizes in the PGM model. Figure \ref{fig:cov_beta_crve} presents the empirical coverage rates for the 95\% confidence intervals of the fixed-effect slope ($\beta$). As a baseline, when effect sizes are perfectly independent within clusters ($r=0$), the naive PGM model performs adequately, maintaining coverage near the nominal 0.95 level. The PGM-CRVE model, however, exhibits undercoverage at smaller cluster sizes in this independent scenario. This behavior is a finite-sample characteristic of the robust sandwich estimator, which relies on asymptotic properties. As the number of clusters increases, the PGM-CRVE coverage naturally aligns with the naive model, stabilizing at the nominal level.

However, as within-cluster correlation increases ($r=0.4$ and $r=0.8$), ignoring the dependency structure leads to severe consequences for statistical inference. The naive PGM model suffers from a catastrophic collapse in coverage for the slope. Because ignoring the clustering structure leads to severely underestimated standard errors, the naive confidence intervals are spuriously narrow. In contrast, the PGM-CRVE approach effectively corrects this structural issue. While the PGM-CRVE model still exhibits undercoverage when the number of clusters is small due to its asymptotic nature, it reliably recovers valid statistical inference and stabilizes at the nominal 0.95 level as the number of clusters grows.

\begin{figure}[htbp]
    \centering
    \includegraphics[width=\textwidth]{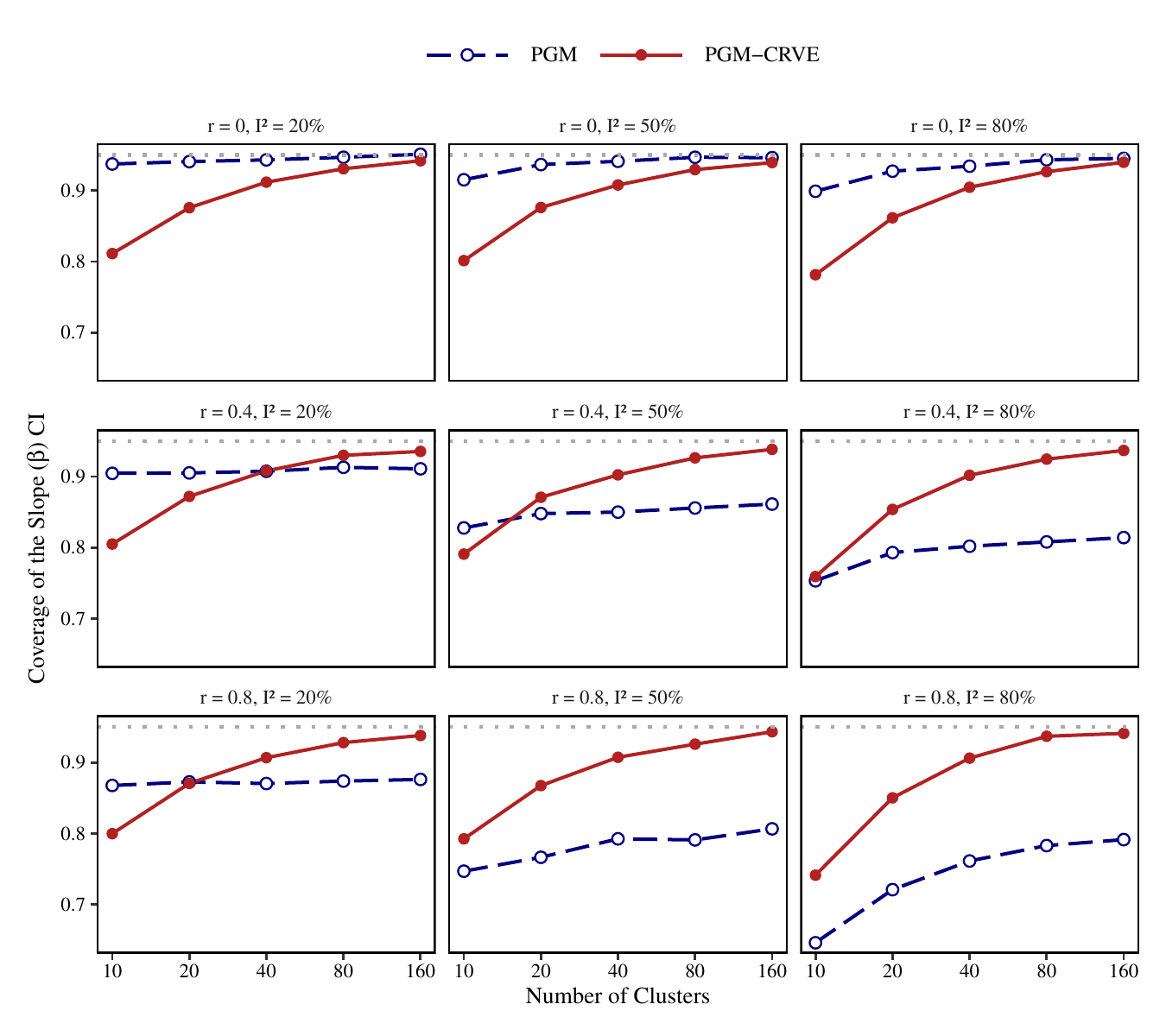}
    \caption{Empirical coverage rates of the 95\% confidence intervals for the fixed-effect slope ($\beta$) across varying degrees of within-cluster correlation ($r$), $I^2$, and number of clusters.}
    \label{fig:cov_beta_crve}
\end{figure}

The consequences of ignoring data dependency are similarly detrimental when estimating complex distributional properties. Figure \ref{fig:cov_prob_crve} displays the average empirical coverage for the 95\% confidence intervals of the conditional tail probabilities. Under moderate to high dependency, the coverage of the naive PGM model for these tail probabilities completely collapses. In the most extreme scenario evaluated, the naive model's coverage plunges to near 0.6, rendering any risk assessment or inference on probabilities highly misleading. Applying the robust empirical sandwich estimator is strictly necessary here to rescue the inference. Consistent with the slope estimates, the PGM-CRVE model's initial coverage is slightly tempered at 10 and 20 clusters due to finite-sample properties. However, by effectively accounting for the correlated data structure, it consistently draws the coverage rates back toward the nominal 0.95 level as the number of independent clusters grows, ensuring robust evaluation of the true effects distribution even in highly dependent structures.

\begin{figure}[htbp]
    \centering
    \includegraphics[width=\textwidth]{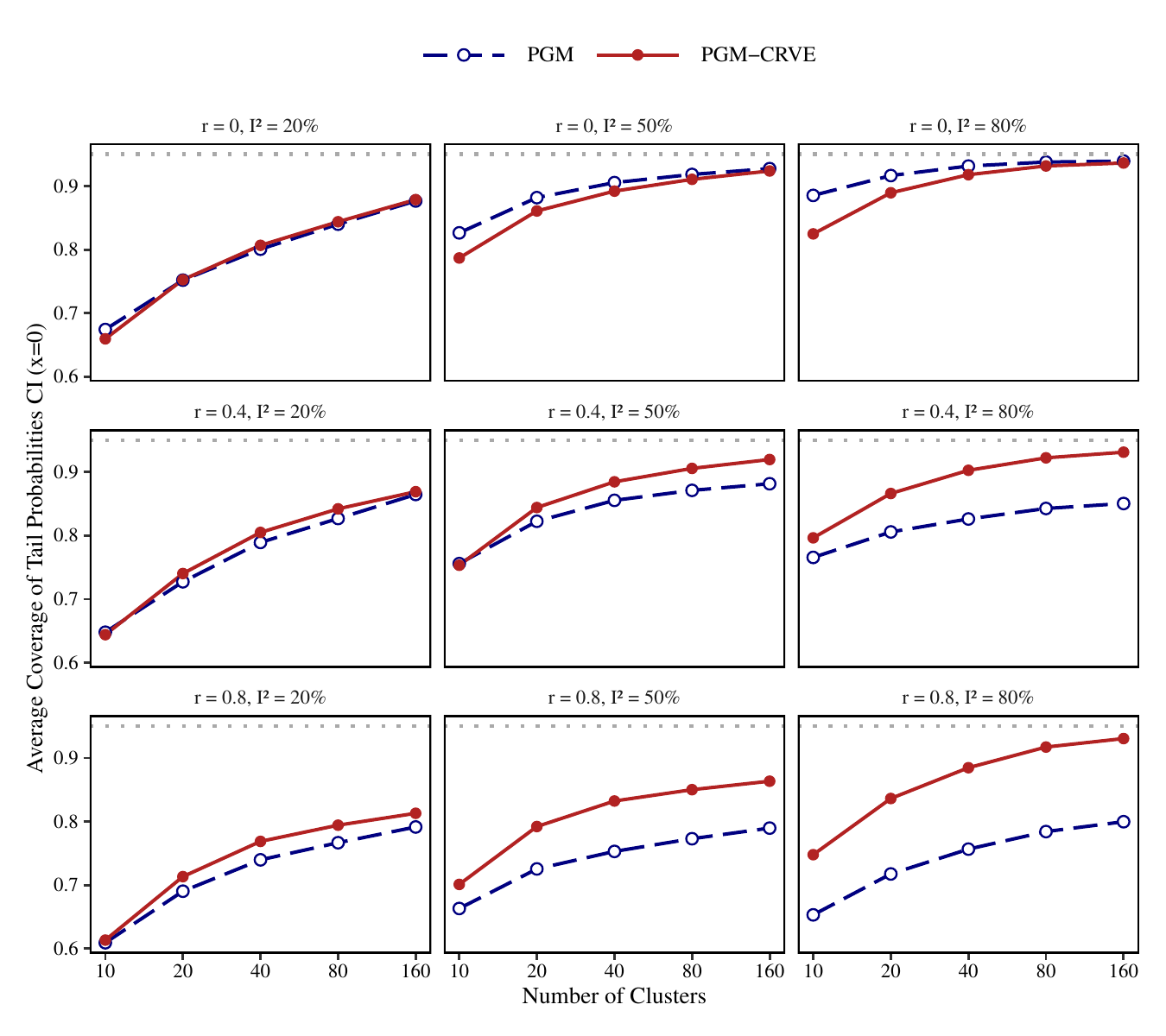}
    \caption{Average empirical coverage rates of the 95\% confidence intervals for the conditional tail probabilities across varying degrees of within-cluster correlation ($r$), $I^2$, and number of clusters.}
    \label{fig:cov_prob_crve}
\end{figure}

\section{Application: Environmental Education}

\subsection{Overall Efficacy}

To illustrate the practical value of the Penalized Gaussian Mixture method in a real-world setting, we applied our framework to a large, heterogeneous dataset from a recent meta-analysis by~\citet{VanDeWetering2022}. This study synthesizes the effects of environmental education (EE) programs on various outcomes in children and adolescents. We focus specifically on the ``Knowledge'' outcome, which comprises a multilevel structure of 231 effect sizes nested within 123 independent studies. The effect sizes represent the standardized mean difference between EE participants and control groups.

We first fit a standard multivariate meta-analysis model with robust variance estimation (RVE) using restricted maximum likelihood (REML) to accommodate the nested data structure. The conventional model estimated an overall mean effect of 0.953 (95\% CI: [0.805, 1.101]). The variance components were estimated at both the effect level ($\sigma^2_1 = 0.157$) and the study level ($\sigma^2_2 = 0.533$), yielding a substantial total heterogeneity variance of approximately 0.690.

Next, we fit the proposed PGM model to the same dataset. To optimize the model's complexity, we performed a grid search over the smoothing parameter $\lambda_\alpha$. Figure~\ref{fig:diagnostic} displays the diagnostic profiles, showing how the Akaike Information Criterion (AIC) and the effective degrees of freedom (EDF) respond to varying levels of penalization. By minimizing the AIC, we selected an optimal smoothing parameter that yielded an EDF of roughly 8.9. Notably, the optimized PGM model achieved an AIC of 513.4, which represents a marked improvement in model fit compared to the standard multivariate model's AIC of 523.3.

\begin{figure}[htbp]
    \centering
    \includegraphics[width=0.8\textwidth]{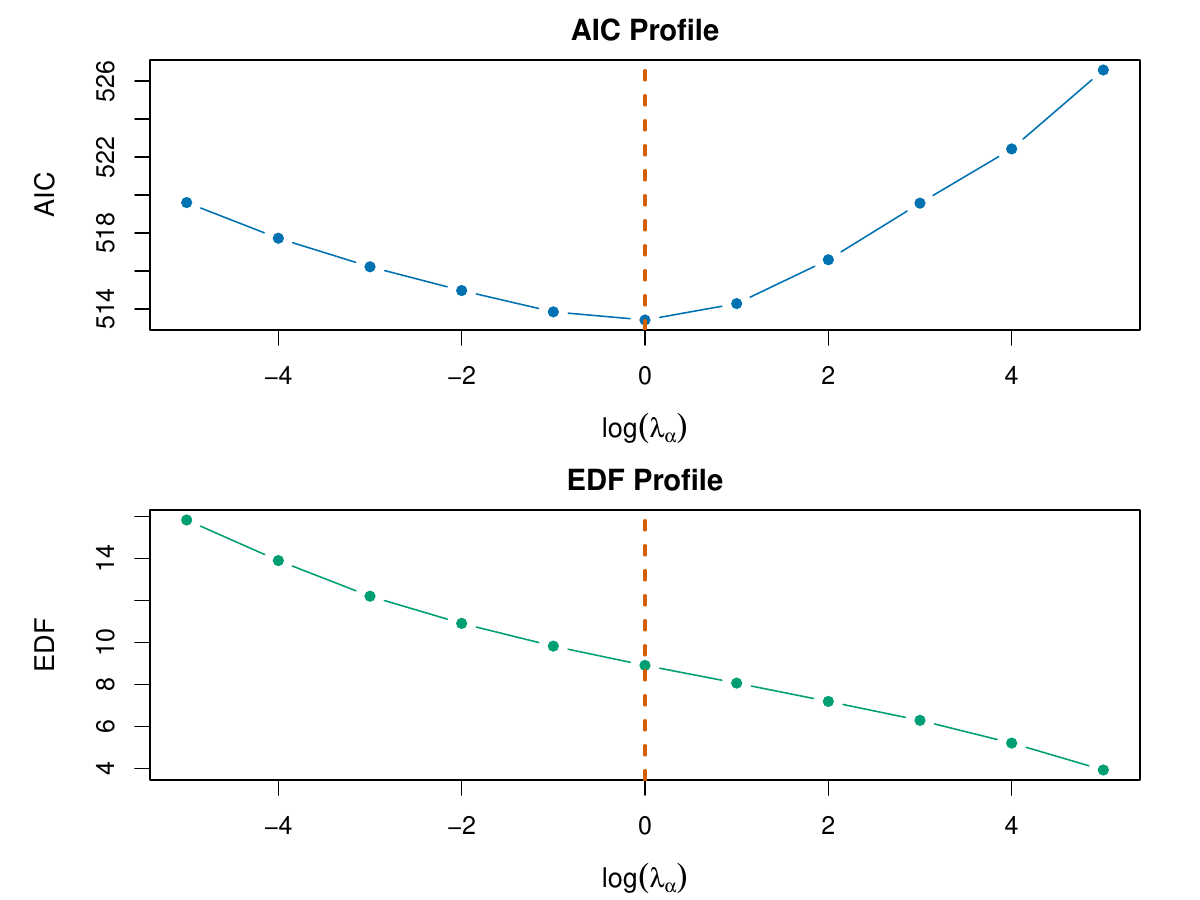}
    \caption{Diagnostic plot for the PGM model, illustrating the selection of the optimal smoothing parameter $\lambda_\alpha$ based on the AIC profile.}
    \label{fig:diagnostic}
\end{figure}

When comparing the central tendencies, both models draw a similar conclusion regarding the average efficacy of the programs. The PGM model estimated an overall mean effect of 0.920 (95\% CI: [0.768, 1.072]) and a total variance of 0.708, which closely mirrors the summary moments derived from the standard model. However, relying solely on these means and variances obscures a profound difference in the underlying distributional shapes, which carries significant implications for interpreting the intervention's real-world reliability.

Figure~\ref{fig:density} visualizes this divergence by superimposing the estimated densities of true effects from both models over the observed effect sizes. The standard normal model (blue dashed line) represents the marginal density, combining both variance components into a single symmetric Gaussian curve. To accommodate the large overall variance---which is primarily driven by a heavy right tail of highly successful programs---the normal model is forced to artificially inflate its left tail. Consequently, the normal model produces a widely dispersed 95\% prediction interval ranging from -0.675 to 2.581. More troublingly, it estimates that the probability of an adverse, negative effect ($P(\theta < 0)$) is approximately 12.6\%. Under this rigid assumption, a policymaker might conclude that environmental education is actually detrimental for a substantial minority of student populations.

\begin{figure}[htbp]
    \centering
    \includegraphics[width=0.9\textwidth]{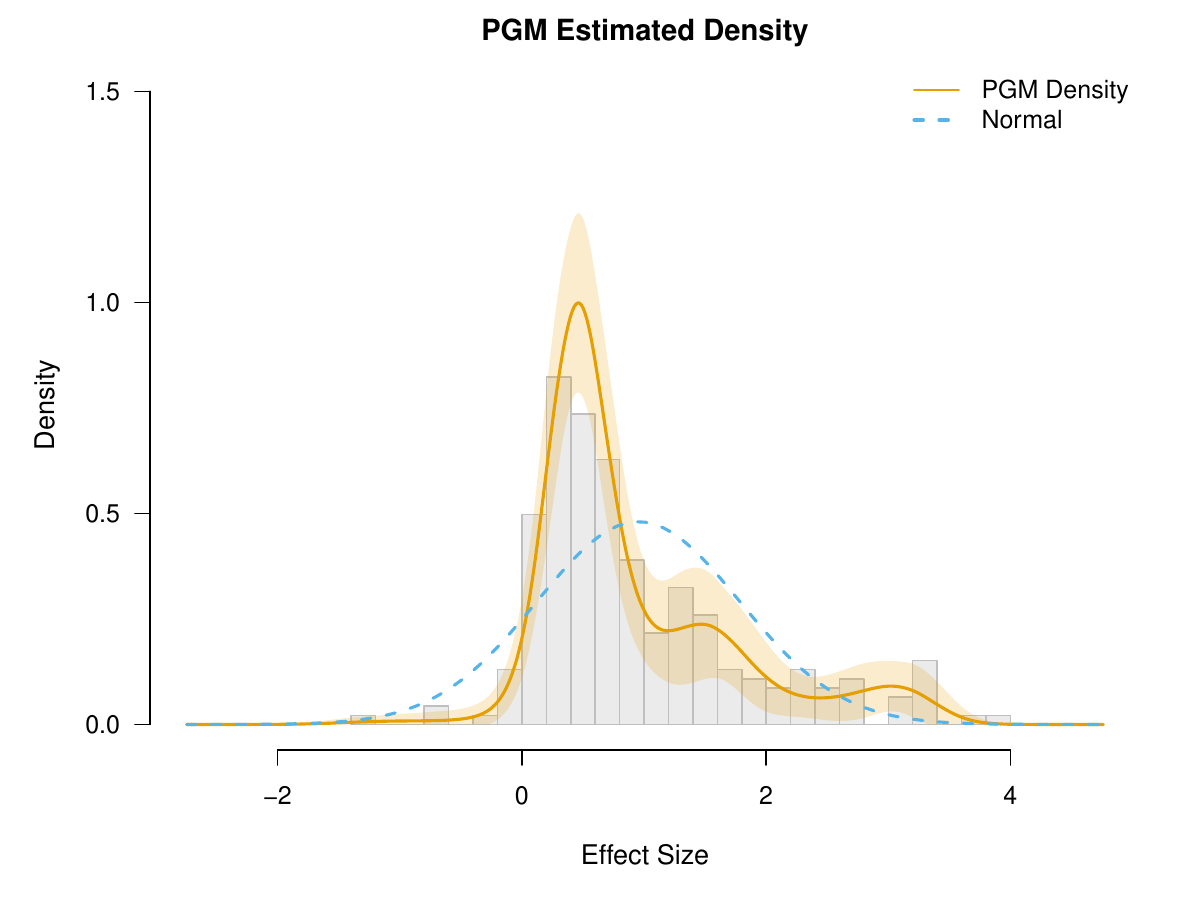}
    \caption{Estimated density of true effects for Environmental Education knowledge outcomes. The histogram shows the observed effect sizes. The blue dashed line represents the marginal density from the standard multivariate normal model, while the solid orange line represents the flexible PGM model with its 95\% pointwise confidence intervals (shaded region).}
    \label{fig:density}
\end{figure}

In stark contrast, the PGM model (solid orange line) adaptively identifies and contours to the true right-skewed nature of the evidence base. It tightens the density curve around zero on the left while extending a long, heavy tail to the right to capture the ``super-responder'' contexts. Because it does not mandate symmetry, the PGM model drastically adjusts our understanding of the intervention's lower bound. Its estimated 95\% prediction interval is notably narrower on the negative end, spanning from -0.091 to 3.156. Most importantly, the estimated probability of a negative effect drops significantly to just 3.9\% (95\% CI: [1.1\%, 6.7\%]). 

This application powerfully demonstrates the risk of default parametric assumptions in social science meta-analyses. While the normal-normal model accurately captured the average benefit, its forced symmetry severely exaggerated the risks. The PGM framework, by contrast, was able to confirm the high average efficacy while simultaneously revealing a much safer, more reliable intervention profile than standard methods would suggest.

\subsection{Explaining Heterogeneity: Formal vs. Non-Formal Settings}

While the intercept-only model provides a comprehensive view of the overall distribution, a primary goal of meta-analysis is to explain the observed heterogeneity. To this end, we examined the impact of the educational setting. Over time, environmental educational programs for children and adolescents have been implemented in formal, school-based settings as well as in non-formal settings, for example at zoos and nature centers. The original study coded this as a binary covariate where $1$ represents formal, school-based settings, and $0$ represents non-formal settings.

We first applied the location-shifting PGM model, which assumes the covariate shifts the entire distribution along the effect-size continuum without altering its fundamental shape. As shown in Figure~\ref{fig:location_shifting_app}, the model estimated a statistically significant negative shift for formal settings ($\hat{\beta} = -0.161$, $SE = 0.082$, $p = 0.048$, 95\% CI: [-0.321, -0.001]). Under this rigid assumption, one might confidently conclude that formal school-based programs are systematically, albeit slightly, less effective than non-formal programs.

\begin{figure}[htbp]
    \centering
    \includegraphics[width=0.9\textwidth]{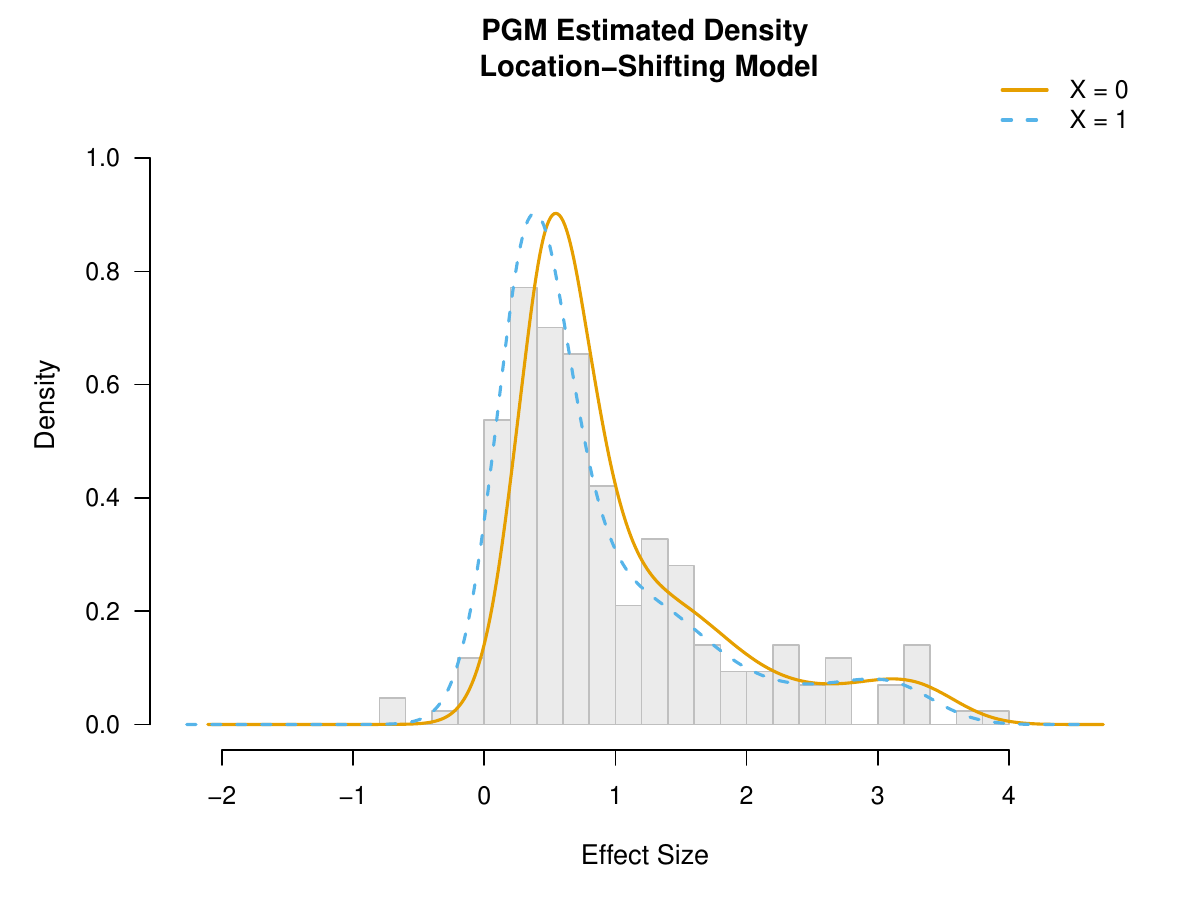}
    \caption{Estimated densities from the Location-Shifting Model. The entire distribution is shifted based on the setting (formal vs. non-formal), assuming the structural shape and variance remain perfectly constant across contexts.}
    \label{fig:location_shifting_app}
\end{figure}

However, applying the shape-morphing model reveals a much more nuanced and informative reality (Figure~\ref{fig:shape_morphing_app}). When allowing the distribution's shape to dynamically vary by setting, we found that the difference in mean effects between the two settings is not actually statistically significant ($\Delta\mu = 0.239$, $SE = 0.166$, $p = 0.149$). Instead, the crucial difference lies in their dispersion and structural shape. Formal settings exhibit drastically higher heterogeneity, with the estimated variance jumping from 0.311 in non-formal settings to 1.012 in formal settings ($\Delta\sigma^2 = 0.701$, $SE = 0.187$, $p < 0.001$).

\begin{figure}[htbp]
    \centering
    \includegraphics[width=0.9\textwidth]{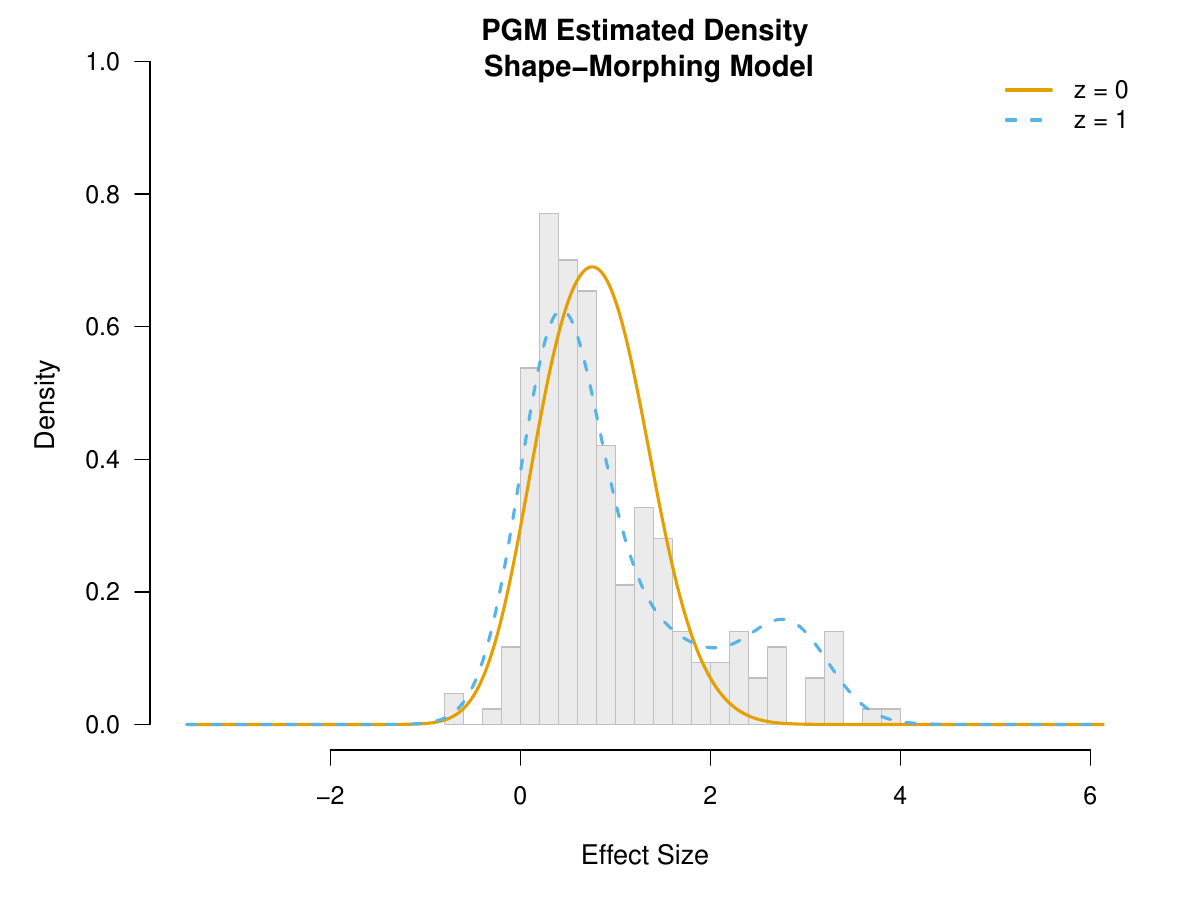}
    \caption{Estimated densities from the Shape-Morphing Model. Unlike the location-shifting approach, this model reveals that formal settings ($z=1$) are characterized by a drastically higher variance and a heavy right tail of ``super-responders'' compared to the more consistent non-formal settings ($z=0$).}
    \label{fig:shape_morphing_app}
\end{figure}

Visually, the shape-morphing plot demonstrates that non-formal settings ($z=0$, solid orange line) produce a relatively consistent, reliable positive effect heavily concentrated around the mean. In contrast, formal school-based settings ($z=1$, dashed blue line) are highly polarized. While a large portion of the mass is situated slightly lower than in non-formal settings, there is a prominent, heavy right tail (and a potential secondary mode) reaching past an effect size of 2.0. This indicates the existence of a subgroup of highly successful ``super-responder'' programs or specific school contexts where the intervention was exceptionally effective. A traditional meta-regression focusing only on mean shifts (like the location-shifting model) would completely obscure this dramatic difference in reliability and the existence of this high-performing subpopulation.

\section{Discussion}

The primary contribution of this paper is the introduction of a flexible framework that allows meta-analysts to extract deeper insights from their data than is possible with conventional methods. While the standard random-effects model compresses the complexity of a dataset into just two summary statistics, the mean and the variance, the Penalized Gaussian Mixture (PGM) approach leverages the full information of the observed effect sizes to reconstruct the entire distribution of true effects. This capability is particularly timely given the evolving landscape of social science research: larger, more heterogeneous datasets, often containing hundreds of studies. In such varying contexts, identifying the shape of the distribution is essential for understanding the nuances of an intervention's impact, identifying scenarios where an effect might be negligible or even harmful, and effectively guiding policy.

Our simulation results provide some guidance on the optimal conditions for applying this method. The PGM framework performs best when the ``signal'' of the distribution is strong relative to the ``noise'' of sampling error. Specifically, the method excels in scenarios characterized by a large number of studies and moderate-to-high $I^2$. Under these conditions, the algorithm has sufficient information to distinguish the features of the true effect distribution (such as skewness or bimodality) from the random sampling variation. Conversely, in small meta-analyses or cases where heterogeneity is negligible, the data may be too sparse to support a flexible density estimation. Therefore, we recommend that researchers employ the PGM method primarily for medium-to-large scale meta-analyses where significant between-study variation is present or suspected.

A critical feature of the PGM method is its robustness and universality. A common concern with flexible models is that they might overfit the data, detecting patterns that do not exist. However, by incorporating a third-order difference penalty, our method is designed to be conservative. As demonstrated in the simulation study (specifically the Normal distribution scenario), when the true effects effectively follow a normal distribution, the PGM estimates smoothly reduce to a shape indistinguishable from the normal curve. This property is vital for applied research: it implies that analysts do not need to pre-test for non-normality to decide whether to use this method. Instead, the PGM framework can serve as a universal tool, or at the very least, a routine sensitivity analysis, that defaults to the standard solution when standard assumptions hold, but automatically reveals complex structures when they exist.

In fact, systematic reviews and meta-analyses are widely regarded as the tip of the evidence pyramid, commanding significant influence over public policy and practice guidelines. It is precisely because these results are held in such high esteem that methodologists and applied researchers bear a heightened obligation to ensure their validity. In this context, retaining the normality assumption solely for analytical convenience, especially when faced with large, heterogeneous datasets, risks masking critical nuances or, worse, propagating misleading conclusions.

We acknowledge that this paper serves primarily as an introduction, demonstrating the feasibility and potential of using flexible models in meta-analysis. Several practical implementation details could be extended further to fully optimize new approaches. For example, our current approach selects the grid range, grid density, and the number of components based on heuristic values. Future work could refine these choices through data-driven algorithms, perhaps calibrating the grid resolution based on the number of effect sizes to ensure optimal balance between resolution and parsimony. In terms of estimation and inference, we have relied heavily on the Delta method for constructing confidence intervals. While analytically convenient, this approach can be dangerous, particularly when sampling distributions are asymmetric. A bootstrap method, while computationally more intensive, may offer a more robust alternative for inference. Furthermore, our current model selection relies heavily on a single criterion, the AIC, to determine the optimal penalty parameters. This process is inherently subject to uncertainty and may occasionally select models that overfit sparse data. Future research should prioritize the development of better, more robust model selection methods specifically tailored for this framework. This could involve exploring alternative information criteria, implementing advanced cross-validation techniques adapted for meta-analytic data, or even employing a fully Bayesian framework to capture model 
uncertainty comprehensively. Finally, to facilitate the adoption of the 
proposed framework and encourage reproducible research, we have made 
the \texttt{metapgm} R package and all code necessary to reproduce these 
analyses publicly available. Readers can access the R package on GitHub at \url{https://github.com/daihesui/metapgm}.

\end{document}